\documentclass[a4paper,11pt]{article}

\usepackage[margin=3.7cm]{geometry}

\usepackage[utf8]{inputenc}
\usepackage{amssymb,amsmath,amsthm,mathrsfs,enumitem,pgfplots}
\usepackage[colorlinks=true,allcolors=blue]{hyperref}
\usepackage[affil-it]{authblk}
\usepgfplotslibrary{fillbetween}

\newcommand{\real}{\mathbb{R}}
\newcommand{\nat}{\mathbb{N}}
\newcommand{\cahy}{\mathcal{S}}
\newcommand{\mes}{\mu}
\newcommand{\tll}{\ll}
\newcommand{\supp}{\operatorname{supp}}
\newcommand{\thck}{I}
\newcommand{\vol}{V}
\newcommand{\tf}{t}

\newcommand{\Lip}{\operatorname{Lip}}
\newcommand{\ima}{\operatorname{Im}}

\newcommand{\ball}{B}
\newcommand{\cpq}{\mathcal{C}}

\theoremstyle{plain}
\newtheorem{thm}{Theorem}[section]
\newtheorem{conj[thm]}{Conjecture}
\newtheorem{lem}[thm]{Lemma}
\newtheorem{prop}[thm]{Proposition}
\newtheorem{cor}[thm]{Corollary}

\theoremstyle{definition}
\newtheorem{defn}[thm]{Definition}
\newtheorem{ex}[thm]{Example}

\newtheorem{rem}[thm]{Remark}

\begin{document}

\title{Time functions on Lorentzian length spaces}
\author[A.~Burtscher and L.~Garc\'ia-Heveling]{Annegret Burtscher and Leonardo Garc\'ia-Heveling}

\maketitle

\begin{abstract}
 In general relativity, time functions are crucial objects whose existence and properties are intimately tied to the causal structure of a spacetime and also to the initial value formulation of the Einstein equations. In this work we establish all fundamental classical existence results on time functions in the setting of Lorentzian (pre-)length spaces (including causally plain continuous spacetimes, closed cone fields and even more singular spaces). More precisely, we characterize the existence of time functions by $K$-causality, show that a modified notion of Geroch's volume functions are time functions if and only if the space is causally continuous, and lastly, characterize global hyperbolicity by the existence of Cauchy time functions, and Cauchy sets. Our results thus inevitably show that no manifold structure is needed in order to obtain suitable time functions.
%\end{abstract}

%%% CLASSIFICATION %%%

\bigskip
{\footnotesize
\noindent 
\emph{\textup{2020} Mathematics Subject Classification:} 
53C23 (primary), %Global geometric and topological methods (à la Gromov); differential geometric analysis on metric spaces
53C50, %Lorentz manifolds, manifolds with indefinite metrics (GLOBAL)
53C75, %Geometric orders, order geometry (in DG)
53C80, %Applications of global differential geometry to the sciences
83C05 (secondary). %Einstein's equations (general structure, canonical formalism, Cauchy problems)

\medskip
\noindent
\emph{Keywords:} general relativity, causality theory, time functions, volume functions, spacetimes, Lorentzian length spaces.}
\end{abstract}

%%%%%%%%%%%%%%%%%%%%%%%%%%%% INTRODUCTION %%%%%%%%%%%%%%%%%%%%%%%%%%%%%%%%%%%%%%%%%%%%%%%%%%%%%

\section{Introduction}

On a smooth spacetime $(M,g)$ a continuous function $t\colon M \to \real$ is called a \emph{time function} if it satisfies
\[
 p < q \implies t(p) < t(q) \ \text{ for all } \ p,q \in M,
\]
where $p < q$ means that there exists a future-directed causal curve from $p$ to $q$, and that $p \neq q$. Time functions play a crucial role in Lorentzian causality theory and Einstein's general theory of relativity.

The study of time functions has a long history in general relativity. Their origin can be traced back to the works of Geroch and Hawking in the late 1960s. Geroch introduced \emph{volume time functions} by normalizing the volume of a spacetime to one, and by defining the time of a point $p$ as the volume of its chronological past $I^-(p)$. In his seminal work \cite{Ger} from 1970, Geroch used these volume functions to characterize global hyperbolicity by the existence of Cauchy surfaces and to obtain a topological splitting. Global hyperbolicity is the strongest and most important causality condition in general relativity. Cauchy surfaces represent the natural sets to pose initial conditions for the Einstein equations (for a self-contained exposition see \cite{Ri}). Moreover, global hyperbolicity and its different characterizations play a crucial role in the singularity theorems of Penrose and Hawking (see \cite[Section 8.2]{HaEl}), the Lorentzian splitting theorems (see \cite{Ba,Es,Ga} and follow-up work) and the formulation of Quantum Field Theory on curved backgrounds \cite[Chapter 4]{BGP}.

Building upon Geroch's idea, Hawking~\cite{Haw} showed that volume functions can be ``smeared out'' to obtain time functions at a significantly lower step on the causal ladder, namely stable causality (see also the work of Minguzzi~\cite{Min,Min2} for the same result via the equivalent notion of $K$-causality and \cite[Figure 20]{Min3} for a depiction of the complete causal ladder). This result contributed to Hawking's program to find the minimal causality conditions that one should impose on a spacetime in order to consider it as physically reasonable. As an in-between result between stable causality and global hyperbolicity, Hawking and Sachs~\cite{HaSa} showed in 1974 that Geroch's volume functions are continuous themselves (that is, without the use of an averaging procedure) precisely when the spacetime is causally continuous. Their proof, however, contained a loophole that was later filled by Dieckmann~\cite{Die2,Die}.

After these foundational works, the question remained whether time functions and Cauchy surfaces can be chosen smooth, rather than just continuous. Despite several attempts by Seifert~\cite{Sei}, Sachs and Wu~\cite{SaWu} and Dieckmann~\cite{Die2,Die3}, this problem remained open for decades. Only in the early 2000s it was firmly established by Bernal and S\'anchez~\cite{BeSa1,BeSa,BeSa3} that a spacetime that admits a continuous time function also admits a smooth one, and that Geroch's topological splitting of globally hyperbolic spacetimes can be promoted to a smooth, orthogonal splitting. Building on these results, S\"amann~\cite{Sae} showed that even if the metric tensor is merely continuous, global hyperbolicity still implies the existence of smooth Cauchy surfaces and time functions (but no orthogonal splitting). This is in spite of the pathological behavior that continuous metrics exhibit, first discovered by Chru\'sciel and Grant~\cite{ChGr}, such as ``causal bubbles'' (failure of the push-up property) and not necessarily open chronological futures/pasts~\cite{GH,GKSS,Li,SoWo}. While their work differs significantly from previous approaches, it was recently established by Chru\'{s}ciel, Grant and Minguzzi~\cite{ChGrMi} that also a family of Geroch's time functions are continuously differentiable for globally hyperbolic $C^{2,1}$-metrics, and that Hawking's time functions can be smoothed out.

A radically different approach to show smoothness of time functions is that of Fathi and Siconolfi~\cite{Fa,FaSi}, which uses weak KAM theory. It has the added advantage that it is formulated for smooth manifolds equipped with a continuous field of closed convex tangent cones, of which Lorentzian manifolds are just one class of examples. Their approach has been developed further in very recent works of Bernard and Suhr~\cite{BeSu,BeSu2} (extending Conley theory to this setting) and Minguzzi~\cite{Min4} (using traditional arguments). Besides the elimination of the need of a Lorentzian metric in the theory of closed (causal) cones, also the smooth manifold structure should be removable to some degree. Indeed, it was already pointed out by Minguzzi~\cite{Min5} that sufficient conditions for the existence of time functions can be obtained for more general topological spaces through Nachbin's theory of closed ordered spaces \cite{Na}.

It remained open, however, if and to what extent general (and in particular, all finer) results about the existence and properties of time functions known for smooth spacetimes rely on the causal and topological structure alone. Our present work answers this question fully in the abstract framework of Lorentzian (pre-) length spaces by characterizing the existence of several types of time functions with different steps on the causal ladder.

The framework of Lorentzian (pre-)length spaces is widely applicable as a broad range of singular and regular ``spacetimes'' appearing in the literature are encompassed, including the above-mentioned closed cone fields and causally plain continuous spacetimes. Introduced by Kunzinger and S\"amann~\cite{KuSa} in 2018, the notion of Lorentzian length spaces makes explicit what is already evident in the early works of Weyl, Penrose and what has also been proposed by various other authors (see, for instance, \cite{BLMS,BoSe,BoSe2,Bus,EPS,KrPe,We,Wo} and \cite[Section 4.2.4]{Min3}), namely that one should treat the causal structure as the most fundamental geometric object in the general theory of relativity. At the same time, an attempt has been made to translate metric geometric techniques \`{a} la Gromov that have revolutionized Riemannian geometry to Lorentzian geometry. In this spirit, Lorentzian (pre-)length spaces are defined as essentially metric spaces equipped with a chronological and a causal order satisfying basic order-theoretic properties such as open chronological futures/pasts and the push-up property (see Section~\ref{prelim} for precise definitions). This general theory of ``spacetime-like'' spaces encompasses basically all previously mentioned low-regularity settings. Despite its youth the Lorentzian (pre-)length framework has already celebrated important successes in the context of causality theory~\cite{ACS,KuSa}, inextendibility results~\cite{AGKS,GKS}, synthetic curvature bounds \cite{CaMo,McC,MoSu} (related to energy conditions in general relativity) and stability \cite{AlBu,KuSt} with respect to the null distance.

Nonetheless, the question of when a time function on Lorentzian (pre-)length space exists has been neglected until now, the only exception being the recent work of Kunzinger and Steinbauer~\cite{KuSt} where it is shown that the existence of certain time functions implies strong causality (the converse, however, is false even in the manifold setting). In the present work, we fill this gap by establishing several sharp existence results. The statements and a discussion of our main results follows.

\subsection*{Main results}
 
In this paper we establish three major results relating the existence (and properties) of time functions to three different steps on the causal ladder, starting from the optimal condition for existence ($K$-causality) and building it up to the top one (global hyperbolicity). All our results in the main body of the paper are obtained for Lorentzian pre-length spaces obeying milder axioms than those required of a Lorentzian length space (all definitions are presented in Section~\ref{prelim} in a self-contained way). Establishing our results in the context of Lorentzian pre-length spaces is important because they are more widely applicable and often sufficient (see also \cite{CaMo,KuSt}). In this introduction, however, we state a simplified version of our theorems in the setting of Lorentzian length spaces (at the end we briefly comment on the pre-length case). Generally it is useful to recall that on a smooth spacetime, the local causal and topological structure is very rigid (all neighborhoods look the same). Our proofs, on the other hand, rely almost exclusively on global arguments. This way, we can reduce the local assumptions to the bare minimum, doing away completely with the manifold structure.

Our first result characterizes the mere existence of time functions on Lorentzian (pre-)length spaces by $K$-causality, generalizing a result of Hawking~\cite{Haw} and Minguzzi~\cite{Min2} for smooth spacetimes (note that $K$-causality is equivalent to stable causality in the smooth setting \cite{Min}).

\begin{thm}
 Suppose $X$ is a second countable, locally compact Lorentzian length space. Then $X$ is $K$-causal if and only if $X$ admits a time function. \label{thm1}
\end{thm}

Here it is crucial that the $K$-relation is closed and transitive (by Definition \ref{Kdef}). Then, as already pointed out by Minguzzi~\cite{Min5}, the general theory of topological ordered spaces yields time functions on $K$-causal spaces. To prove the converse statement, namely that existence of a time function implies $K$-causality, we closely follow the approach of Minguzzi~\cite{Min2}, which requires the use of limit curve theorems.

Our second result is concerned with an explicit construction of time functions on Lorentzian (pre-)length spaces that can be equipped with Borel probability measures. Here, we are influenced by Geroch's notion of volume functions \cite{Ger} as well as Hawking's averaging procedure \cite{Haw}. Since the boundaries of light cones, however, are no longer hypersurfaces with measure zero as in the smooth manifold setting, the definition of our \emph{averaged volume functions} as well as the proof of their causal properties and continuity are significantly more involved. The result we obtain is essentially a generalization of a theorem by Hawking and Sachs~\cite{HaSa} (and Dieckmann's rigorous follow-up work \cite{Die}), and thus the corresponding step on the causal ladder is that of causal continuity.

\begin{thm} \label{thm2}
 Let $X$ be a second countable, locally compact Lorentzian length space. Then $X$ is causally continuous if and only if the averaged volume functions on $X$ are time functions.
\end{thm}

While the assumption on second countability of the underlying metric space is crucial in order to have a suitable measure at hand, local compactness can be removed by using a weaker (but in the smooth case equivalent) notion of causal continuity.

Our third result characterizes globally hyperbolic Lorentzian (pre-)length spaces by the existence of Cauchy time functions, whose level sets are Cauchy sets that are intersected by every inextendible causal curve exactly once. The smooth spacetime analogue is the seminal 1970 result of Geroch~\cite{Ger}.

\begin{thm} \label{thm3}
 Let $X$ be a second countable Lorentzian length space with a proper metric structure. Then the following are equivalent:
 \begin{enumerate}
  \item $X$ is globally hyperbolic,
  \item $X$ is non-totally imprisoning and the set of causal curves between any two points is compact,
  \item $X$ admits a Cauchy set,
  \item $X$ admits a Cauchy time function.
 \end{enumerate}
\end{thm}

To establish Theorem~\ref{thm3} we utilize our averaged volume functions introduced already for Theorem~\ref{thm2}, as well as the behavior of inextendible causal curves. The use of \emph{averaged} volume functions poses an additional difficulty in our proofs, compared to the smooth case. The other main challenge is the fact that our Cauchy sets are not hypersurfaces, and in fact very little can be deduced about their topology (hence the name Cauchy \emph{set} instead of surface, see also our discussion below).

The above theorems follow immediately from their sharper versions, Theorems \ref{Kthm}, \ref{ccthm} and \ref{ghthm}, which are obtained for Lorentzian pre-length spaces. Some essential conditions (which are part of the axioms of Lorentzian length spaces), however, still need to be assumed. In Theorem~\ref{Kthm} (generalizing Theorem~\ref{thm1}), for instance, we need to additionally impose the existence of causal curves and their limit curves. The chronological relation, however, is not needed in the proof at all. The proof of Theorem~\ref{ccthm} (corresponding to Theorem~\ref{thm2}), on the other hand, does not require causal curves, but does make use of both the causal and chronological relation. Furthermore, the two relations need to satisfy a compatibility condition that we call ``approximating'', which simply means that the causal futures and pasts are contained in the closure of the chronological ones. Finally, Theorem~\ref{ghthm} (the pre-length version of Theorem~\ref{thm3}) builds upon Theorem~\ref{ccthm}, hence the ``approximating'' condition is also needed here. Moreover, causal curves are used already in the definition of Cauchy set and Cauchy time function, and the limit curve theorems will be important again.

\subsection*{Discussion and outlook}

Since its inception, the framework of Lorentzian length spaces has seen a rapid expansion \cite{ACS,AGKS,CaMo,GKS,KuSa,KuSt}. Notably, Cavalletti and Mondino~\cite{CaMo} recently introduced a notion of Ricci curvature bounds for Lorentzian pre-length spaces, based on optimal transport theory, that mimics the strong energy condition of general relativity, and implies a version of Hawking's singularity theorem. Our work completes another important milestone in establishing the potential of this non-smooth causal theory by fully characterizing the existence of time functions in terms of the causal ladder. As an immediate application, we have now unambiguously established when one can make use of time functions to define the null distance of Sormani and Vega~\cite{SoVe}. Very recently, this notion (as well as convergence) has been investigated by Kunzinger and Steinbauer~\cite{KuSt} in the context of Lorentzian pre-length spaces (certain limits of examples in \cite[Section 5]{AlBu} are also of this type). With our new characterizations, in particular, of global hyperbolicity via Cauchy time functions, more refined convergence/stability results may be obtained. Moreover, it should be straightforward to carry over the cosmological time function of Andersson, Galloway and Howard~\cite{AGH} to the Lorentzian length space framework using the time separation function in place of the Lorentzian distance.

While in broad terms we show that the classical results about time functions admit direct generalizations for Lorentzian length spaces, we also find interesting and not so subtle differences. In particular, although global hyperbolicity is also characterized by the existence of Cauchy sets, these Cauchy sets need not be homeomorphic to each other. This is in stark contrast to the case of spacetimes, where Geroch's celebrated splitting theorem~\cite{Ger} (later refined by Bernal and S\'anchez~\cite{BeSa1}) shows that all Cauchy surfaces on a given spacetime must have the same topology. In fact, Geroch already showed in \cite{GerTopo} that transitions between compact spatial topologies not only contradict global hyperbolicity, but in fact even violate the most basic of all assumptions, chronology. While time travel is a no-go in any physically sound theory, Sorkin~\cite{Sor} argues that topology change is a necessary feature of \emph{any} convincing candidate theory of quantum gravity. Going beyond the setting of smooth Lorentzian manifolds is thus a necessity in order to admit topology change without violating chronology, a common approach being the use of degenerate Lorentzian metrics \cite{BDGSS,Hor}. Current proposals for quantum gravity also predict that physical spacetimes are represented by non-manifold-like structures at small scales, such as causal sets \cite{Surya}, causal dynamical triangulations \cite{Loll}, causal fermion systems \cite{Fin}, or spin foams \cite{NiPe}. An interesting and important next step will be to see if and how the framework of Lorentzian length spaces fits into these quantum gravitational theories.

\subsection*{Outline}

The paper is structured as follows. In Section~\ref{prelim}, we give a self-contained account of the relevant aspects of the theory of Lorentzian (pre-)length spaces, drawing from the existing literature but also introducing new material. We then prove Theorems~\ref{thm1}, \ref{thm2} and \ref{thm3} (more precisely, the corresponding sharper Lorentzian pre-length space versions) in Sections~\ref{Kcausality}, \ref{volfcts} and \ref{ghsec}, respectively.

%%%%%%%%%%%%%%%%%%%%%%%%%%%%%% LORENTZIAN PRE-LENGTH SPACES %%%%%%%%%%%%%%%%%%%%%%%%%%%%%%%%%%%%%%%

\section{Lorentzian pre-length spaces} \label{prelim}

In this section we recall, and partly refine, the definition of Lorentzian (pre-)length spaces, their causality conditions and the limit curve theorems. We use the notation and results of \cite{ACS,KuSa}. A reader familiar with the smooth case will find that most classical concepts are defined in the same way for Lorentzian pre-length spaces (with the difference that some important properties do not follow automatically but have to be imposed separately, such as causal curves themselves).

In Section~\ref{ssec:def} we recall the definition of Lorentzian pre-length space, and of causal curve. This is standard material, except that we use a more precise nomenclature for inextendible causal curves (Definition~\ref{definex}). In Section~\ref{ssec:limitcurves} we revisit the limit curve theorems of Kunzinger and S\"amann~\cite{KuSa} in the slightly weaker framework of ``local weak causal closedness'' following a suggestion of Ak\'e et al.~\cite{ACS}. We also introduce the new notion of \emph{Lorentzian pre-length spaces with limit curves}, which encompasses all necessary assumptions needed for the application of the limit curve theorems. In Section~\ref{ssec:approx} we introduce the notion of \emph{approximating Lorentzian pre-length space}, which can be seen as a much weaker version of Kunzinger and S\"amann's ``localizability''. Most importantly, we show that our newly introduced properties are, in particular, satisfied by all Lorentzian length spaces. Finally, in Section~\ref{ssec:causalitytime}, we define time functions and introduce some elements of causality theory for Lorentzian pre-length spaces. Notably, we give some new characterizations of non-total imprisonment, both in terms of causal curves and of time functions. Additional causality conditions, including $K$-causality and global hyperbolicity, are introduced in later sections when needed.

\subsection{Basic definitions and properties}\label{ssec:def}

\begin{defn}[{\cite[Definition 2.1]{KuSa}}] \label{causet}
 A \emph{causal set} is a set $X$ equipped with a preorder $\leq$ (called \emph{causal} relation) and a transitive relation $\tll$ (called \emph{chronological} or \emph{timelike} relation) contained in $\leq$. 
\end{defn}

The following notation for the timelike/causal future or past of a point is standard
\begin{align*}
 I^+(p) &:= \left\{ x \in X \mid p \tll x \right\}, & J^+(p) &:= \left\{ x \in X \mid p \leq x \right\}, \\
 I^-(p) &:= \left\{ x \in X \mid x \tll p \right\}, & J^-(p) &:= \left\{ x \in X \mid x \leq p \right\}, \\
 I(p,q) &:= I^+(p) \cap I^-(q), & J(p,q) &:= J^+(p) \cap J^-(q),
\end{align*}
and we write $p < q$ if $p \leq q$ and $p \neq q$.

\begin{defn}[{\cite[Definition 2.8]{KuSa}}]\label{lpls}
 A \emph{Lorentzian pre-length space} is a causal set $(X,\tll,\leq)$ equipped with a metric $d$ and a lower semicontinuous function $\tau\colon X \times X \to [0,\infty]$ satisfying, for all $x,y,z \in X$
 \begin{enumerate}
  \item $\tau(x,y) > 0$ $\Longleftrightarrow$ $x \tll y$,
  \item $\tau(x,y) = 0$ if $x \not\leq y$,
  \item $\tau(x,z) \geq \tau(x,y) + \tau(y,z)$ if $x \leq y \leq z$.
 \end{enumerate}
\end{defn}

Occasionally we denote a Lorentzian pre-length space simply by $X$. It follows from Definition~\ref{lpls} that the sets $I^\pm(p)$ are open for all $p \in X$, a fact that we will also refer to as the \emph{openness of $\tll$}. The crucial push-up property extends from the smooth situation.

\begin{lem}[{Push-up \cite[Lemma 2.10]{KuSa}}] \label{push-up}
 Let $(X,d,\tll,\leq,\tau)$ be a Lorentzian pre-length space and $x,y,z \in X$ with $x \tll y \leq z$ or $x \leq y \tll z$. Then $x \tll z$.
\end{lem}

 The function $\tau$ in Definition~\ref{lpls} is often called \emph{time separation function} or \emph{Lorentzian distance function}. We will not use this terminology. In fact, we never need the function $\tau$ by itself but just the openness of $\tll$ and the push-up property (we could also trivially set $\tau(p,q) = \infty$ if $p \tll q$ and $=0$ otherwise).

\medskip
In smooth Lorentzian geometry the causal character of curves determines the timelike and causal future and past, that is $I^\pm$ and $J^\pm$, respectively. In Lorentzian pre-length spaces it is the other way round.

% \comment[id=A]{Think about curves vs.\ paths, could be improved but should be consistent with literature.}

\begin{defn}[{\cite[Definition 2.18]{KuSa}}] \label{defcauscurv}
 Let $(X,d,\tll,\leq,\tau)$ be a Lorentzian pre-length space and $I$ be any (open, half-open, or closed) interval in $\real$. A non-constant locally Lipschitz path $\gamma \colon I \to X$ is called a
 \begin{enumerate}
  \item \emph{future-directed causal curve} if $\gamma(s_1) \leq \gamma(s_2)$ for all $s_1 < s_2 \in I$.
  \item \emph{past-directed causal curve} if $\gamma(s_2) \leq \gamma(s_1)$ for all $s_1 < s_2 \in I$.
 \end{enumerate}
 Future- and past-directed \emph{timelike} curves are defined analogously by replacing $\leq$ with $\tll$.
\end{defn}

\begin{rem}
 By a result in metric geometry (see, for instance, \cite[Proposition 2.5.9]{BBI}), we can parametrize any causal curve by $d$-arclength (the reference is for closed intervals only, but the proof is in fact valid for any interval). Recall that $\gamma\colon I \to X$ is parametrized by $d$-arclength iff
 \[ L^d(\gamma \vert_{[a,b]}) = b - a \text{ for all } [a,b] \subseteq I.\]
Since a curve that is parametrized by $d$-arclength is automatically $1$-Lipschitz continuous, causal curves remain causal when parametrizing them by $d$-arclength. Hence, without loss of generality, we can assume that causal curves are parametrized in a way so that they are not locally constant, i.e., not constant on any open subinterval of $\real$.
\end{rem}

\begin{defn}[{\cite[Definition 3.1]{KuSa}}] \label{path-connected}
 A Lorentzian pre-length space $X$ is called \emph{causally path-connected} if for every $p < q$ there exists a future-directed causal curve connecting $p$ and $q$, and for every $p \tll q$ a future-directed timelike curve connecting $p$ and $q$.
\end{defn}

\begin{defn}[{\cite[Definition 2.19]{ACS}}]\label{lwcc}
For a subset $U$ of a Lorentzian pre-length space $X$ we define the relation $\leq_U$ by
 \[
  p \leq_U q :\Longleftrightarrow \text{there is a future-directed causal curve from $p$ to $q$ in $U$.} 
 \]
 A neighborhood $U$ is called \emph{weakly causally closed} if $\leq_U$ is closed, and the Lorentzian pre-length space $X$ is called \emph{locally weakly causally closed} if every point $p \in X$ is contained in a weakly causally closed neighborhood $U$.
\end{defn}

Definition~\ref{lwcc} is satisfied on \emph{any} smooth Lorentzian manifold, and thus acts as a replacement for regularity on a Lorentzian pre-length space. In contrast, the ``local causal closedness'' condition of Kunzinger and S\"amann~\cite[Definition 3.4]{KuSa} is stronger than Definition~\ref{lwcc} because it requires that $\leq$ restricted to $\overline U \times \overline U$ is closed. For instance, in the smooth case the latter notion would only be satisfied on strongly causal spacetimes (see \cite[p.~6]{ACS} for a detailed discussion). Note also that weak causal closedness is most natural on causally path-connected spaces (as it would be a void condition on spaces with no causal curves at all), so one may even include causal path-connectedness in the definition, as done implicitly in \cite{ACS}.

\medskip
Finally, we refine the concept of an inextendible curve\footnote{Doubly-inextendible causal curves are often simply called ``inextendible''. On the other hand, Kunzinger and S\"amann \cite{KuSa} call a curve inextendible if it is either future- or past-inextendible (or both), and have no need for the concept of double-inextendibility. We will be more precise when needed.}.

\begin{defn} \label{definex}
 Let $X$ be a Lorentzian pre-length space and $\gamma\colon (a,b) \to X$ be a future-directed causal curve. If there exists a causal curve $\bar \gamma\colon (a,b] \to X$ such that $\bar \gamma \vert_{(a,b)} = \gamma$, we say that $\gamma$ is \emph{future-extendible}. If there exists a causal curve $\tilde \gamma\colon [a,b) \to X$ such that $\tilde \gamma \vert_{(a,b)} = \gamma$, we say that $\gamma$ is \emph{past-extendible}. We say that $\gamma$ is \emph{future-(past-)inextendible} if it is not future-(past-)extendible, and \emph{doubly-inextendible} if it is neither future- nor past-extendible.
\end{defn}

The analogous definition for past-directed causal curves is obtained by interchanging future and past in Definition~\ref{definex}. The definition applies accordingly to half-open intervals. If a path is defined on all of $\real$, we mean extendibility to $\pm \infty$. Alternatively we can parametrize it by arclength and apply the following lemma (which, of course, admits also a past version).

\begin{lem} \label{lem312}
 Let $X$ be a locally weakly causally closed Lorentzian pre-length space, let $-\infty < a < b \leq \infty$ and let $\gamma \colon [a, b) \to X$ be a future-directed causal curve parametrized with respect to $d$-arclength. If $(X, d)$ is a proper metric space or the curve $\gamma$ is contained in a compact set, then $\gamma$ is future-inextendible if and only if $b = \infty$. In this case $L^d (\gamma) = \infty$. Moreover, $\gamma$ is future-inextendible if and only if $\lim_{t \nearrow b} \gamma(t)$ does not exist.
\end{lem}

\begin{proof}
 First, assume that $b = \infty$. If $\gamma$ admitted an extension $\bar\gamma$ as in Definition~\ref{definex}, then $\bar\gamma$ would be locally Lipschitz and have two endpoints. Therefore, $L^d(\bar\gamma) < \infty$, but also $L^d(\bar\gamma)=L^d(\gamma) = b-a$ (since we have only added one point), a contradiction.
 The rest of the proof is the same as \cite[Lemma 3.12]{KuSa}. Note that there, the assumption of local ``strong'' causal closedness is only applied to points which lie on $\gamma$. Hence that proof also works with our notion of weakly causally closed neighborhood.
\end{proof}

In the remaining subsection we recall the definition of Lorentzian length space (including necessary preliminary notions) as introduced in \cite{KuSa}. While we will not directly work with Lorentzian length spaces in the main body of this paper, our results about pre-length spaces immediately also lead to useful Corollaries in this setting (see Introduction). A Lorentzian length space is, in essence, just the Lorentzian analogue of length metric spaces generalizing Riemannian manifolds where the time separation function $\tau$ is used in place of a distance function to measure lengths and the admissible class of curves respects causality. More precisely, if $\gamma \colon [a,b] \to X$ is a future-directed causal curve, then its \emph{$\tau$-length} $L_{\tau}(\gamma)$ is defined by (see \cite[Definition 2.24]{KuSa})
\[
 L_{\tau}(\gamma) := \inf \left\{ \sum_{i=0}^{N-1} \tau(\gamma(t_i), \gamma( t_{i+1})) \; \bigg| \; %\mid  
 a = t_0 < t_1 < ... < t_N = b, N \in \nat \right\}.
\]
In addition, the notion of localizability is needed.

\begin{defn}[{\cite[Definition 3.16]{KuSa} and \cite[Definition 2.22]{ACS}\footnote{Similarly to the difference between weak and ``strong'' local causal closedness, there is a difference between the notions of localizability in \cite[Definition 3.16]{KuSa} and \cite[Definition 2.22]{ACS} regarding the meaning of $\tll_{U_p}, \leq_{U_p}$.}}] \label{defloc}
 We call a Lorentzian pre-length space $(X,d,\tll,\leq,\tau)$ \emph{localizable} if for every point $p \in X$, there exists a neighborhood $U_p$ of $p$ such that
 \begin{enumerate}
  \item There exists a constant $C > 0$ such that for  all causal curves contained in $U_p$ we have $L^d(\gamma)\leq C$.
  \item For every $q \in U_p$ we have $I^\pm(q) \cap U_p \neq \emptyset$.
  \item There exists a continuous function $\omega_p \colon U_p \times U_p \to [0,\infty)$ such that $(U_p,d \vert_{U_p \times U_p}, \\ \tll_{U_p}, \leq_{U_p}, \omega_p)$ is a Lorentzian pre-length space. Moreover, for all $x,y \in U_p$ with $x < y$, it holds that
  \[
   \omega_p(x,y) = \max \{L_\tau(\gamma) \mid \gamma \colon [a,b] \to U_p \text{ future-dir. causal from $x$ to $y$ } \},
  \]
  so in particular there exists a maximizing causal curve between $x$ and $y$.
 \end{enumerate}
\end{defn}

 In the main sections of this paper, only assumptions (i) and (ii) of Definition \ref{defloc} are needed, thus we restate them separately in the upcoming sections (see Definitions~\ref{defdcomp} and \ref{approx}, Lemma~\ref{pathapprox} and Proposition~\ref{LLSareapprox}). By further assuming that also $\tau$ is given by length-maximization (but without necessarily requiring the existence of global maximizers), one obtains a Lorentzian length space.

\begin{defn}[{\cite[Definition 3.22]{KuSa}}] \label{def:lls}
 A causally path-connected, locally (weakly) causally closed and localizable Lorentzian pre-length space $(X,d,\tll,\leq,\tau)$ is called a \emph{Lorentzian length space} if for all $p,q \in X$
 \[
 \tau(p,q) = \sup \{L_\tau(\gamma) \mid \gamma \text{ future-directed causal from $p$ to $q$} \}.
 \]
\end{defn}

\subsection{Limit curve theorems}\label{ssec:limitcurves}

We revisit the limit curves theorems of Kunzinger and S\"amann~\cite[Section 3.2]{KuSa} and relax their assumption of local causal closedness to local weak causal closedness (see Definition~\ref{lwcc}). That this extension is possible was already pointed out by Ak\'e et al.~\cite[p.\ 8]{ACS}. The limit curve theorems are crucial for Sections~\ref{Kcausality} and \ref{ghsec}.

\medskip
We start with \cite[Lemma 3.6]{KuSa} where, instead of pointwise convergence, we need to assume locally uniform convergence.

\begin{lem} \label{lem36}
 Let $X$ be a causally path-connected locally weakly causally closed Lorentzian pre-length space and let $(\gamma_n)_{n}$ be a sequence of future-directed causal curves $\gamma_n \colon I \to X$ converging locally uniformly to a non-constant locally Lipschitz curve $\gamma \colon I \to X$. Then $\gamma$ is future-directed causal.
\end{lem}

\begin{proof}
 For every $s \in I$ there exists a weakly causally closed neighborhood $U_{\gamma(s)}$. By continuity, we can pick $s_1 < s < s_2$ such that $\gamma([s_1,s_2]) \subseteq U_{\gamma(s)}$ (if $s$ is a boundary point of $I$, then $s_1=s$ or $s_2=s$ is chosen). Assume additionally that we choose $s_1, s_2$ close enough such that $(\gamma_n)_{n}$ converges uniformly on $[s_1,s_2]$. This implies that there exists $n_0 \in \nat$ such that for all $n > n_0$, $\gamma_n([s_1,s_2]) \subseteq U_{\gamma(s)}$. Now it follows from the definition of weakly causally closed neighborhood that $\gamma$ restricted to $[s_1,s_2]$ is future-directed causal. Since $s$ was arbitrary, we can decompose $\gamma$ as a concatenation of future-directed causal curves, hence by transitivity of $\leq$, $\gamma$ is future-directed causal on $I$.
\end{proof}

\begin{thm}[Limit curve theorem] \label{thm37}
 Let X be a causally path-connected locally weakly causally closed Lorentzian pre-length space. Let $(\gamma_n)_n$ be a sequence of future-directed causal curves $\gamma_n \colon [a, b] \to X$ that are uniformly Lipschitz continuous, i.e., there is an $L > 0$ such that $\Lip(\gamma_n ) \leq L$ for all $n \in \nat$. Suppose that there exists a compact set that contains every $\gamma_n$ or that $d$ is proper and that the curves $(\gamma_n)_n$ accumulate at some point, i.e., there is a $t_0 \in [a, b]$ such that $\gamma_n (t_0 ) \to x_0 \in X$. Then there exists a subsequence $(\gamma_{n_k})_k$ of $(\gamma_n)_n$ and a Lipschitz continuous curve $\gamma \colon [a, b] \to X$ such that $\gamma_{n_k} \to \gamma$ uniformly. If $\gamma$ is non-constant, then $\gamma$ is a future-directed causal curve.
\end{thm}

\begin{proof}
 The proof of \cite[Theorem 3.7]{KuSa} goes through. The assumption of local ``strong'' causal closedness is only used to invoke \cite[Lemma 3.6]{KuSa}, but since the convergence is uniform, we can replace it by our Lemma \ref{lem36}.
\end{proof}

This first limit curve theorem is already very useful. In order to formulate our second limit curve theorem, we need a better control over the $d$-length of causal curves.

\begin{defn}[{\cite[Definition 3.13]{KuSa}}] \label{defdcomp}
 A Lorentzian pre-length space $(X,d,\allowbreak \tll,\leq,\tau)$ is called \emph{$d$-compatible} if for every $p \in X$ there exists a neighborhood $U$ of $p$ and a constant $C > 0$ such that $L^d(\gamma) \leq C$ for all causal curves $\gamma$ contained in $U$.
\end{defn}

To ease the nomenclature, we group some of our previous assumptions into the following definition.

\begin{defn} \label{llslimitcurves}
 A \emph{Lorentzian pre-length space with limit curves} is a causally path-connected, locally weakly causally closed, and $d$-compatible Lorentzian pre-length space.
\end{defn}

It is an easy consequence that Lorentzian length spaces are particular cases of Lorentzian pre-length spaces with limit curves (see also Proposition~\ref{LLSareapprox} below).

\begin{thm}[Limit curve theorem for inextendible curves] \label{thm314}
 Let $X$ be a Lorentzian pre-length space with limit curves. Let $(\gamma_n )_n$ be a sequence of future-directed causal curves $\gamma_n \colon [0, L_n ] \to X$ which are parametrized with respect to $d$-arclength and satisfy $L_n := L^d (\gamma_n ) \to \infty$. If there exists a compact set that contains every curve $\gamma_n ([0, L_n ])$ or if $d$ is proper and $\gamma_n (0) \to x$ for some $x \in X$, then there exists a subsequence $(\gamma_{n_k})_k$ of $(\gamma_n )_n$ and a future-directed causal curve $\gamma \colon [0, \infty) \to X$ such that $\gamma_{n_k} \to \gamma$ locally uniformly. Moreover, $\gamma$ is future-inextendible.
\end{thm}

\begin{proof}
 The proof of \cite[Theorem 3.14]{KuSa} goes through. The assumption of local ``strong'' causal closedness is only used to invoke \cite[Lemmas 3.6 and 3.12, Theorem 3.7]{KuSa}, so we can replace them by our Lemma \ref{lem36}, Lemma \ref{lem312} and Theorem \ref{thm37} respectively. 
\end{proof}

While the limit curve theorems are stated for future-directed curves, they of course also hold for past-directed ones.

\subsection{Approximating Lorentzian pre-length spaces}\label{ssec:approx}

In this subsection we introduce our new ``approximating'' condition relating the causal structure and the topology on $X$. It is satisfied on all spacetimes regardless of their place in the causal ladder, and will be crucial in Section~\ref{volfcts}. In Proposition~\ref{LLSareapprox} we show that all Lorentzian length spaces automatically fulfill the ``approximating'' condition, and also our earlier Definition~\ref{llslimitcurves}.

\begin{defn} \label{approx}
 A Lorentzian pre-length space $(X,d,\tll,\leq,\tau)$ is called \emph{approximating} if for all points $p \in X$ it holds that $J^\pm(p) \subseteq \overline{I^\pm (p)}$.
 
 It is called \emph{future-(past-)approximating} if the approximating property holds for $+$($-$).
\end{defn}

The approximating property can equivalently be characterized via sequences as follows.

\begin{lem}\label{lemapprox}
 Let $(X,d,\tll,\leq,\tau)$ be a Lorentzian pre-length space. Then $X$ is (future-/past-)approximating if and only if for every point $p \in X$ there exists a sequence $(p^\pm_n)_n$ in $I^\pm(p)$ such that $p^\pm_n \to p$ as $n\to\infty$.
\end{lem}

We say that the sequence $(p^+_n)_n$ approximates $p$ from the future, and that the sequence $(p^-_n)_n$ approximates $p$ from the past.

\begin{proof}
 That such sequences exist on approximating spaces is obvious, because $p \in J^\pm(p) \subseteq \overline{I^\pm(p)}$. To show the converse, suppose $q \in J^+(p)$ and $(q_n^+)$ is a sequence in $I^+(q)$ approximating $q$ from the future. By the push-up Lemma~\ref{push-up}, $q_n^+ \in I^+(p)$ for all $n \in \nat$, hence $q \in \overline{I^+(p)}$.
\end{proof}

Assuming that $X$ is causally path-connected, we get even more characterizations, which in the smooth case are in fact the most widely used ones.

\begin{lem} \label{pathapprox}
 Let $(X,d,\tll,\leq,\tau)$ be a causally path-connected Lorentzian pre-length space. Then the following are equivalent:
 \begin{enumerate}
     \item $X$ is future-(past-)approximating,
     \item $I^+(p) \neq \emptyset$ ($I^-(p) \neq \emptyset$) for all $p \in X$,
     \item for every point $p \in X$ there exists a future-(past-)directed timelike curve $\gamma \colon [a,b) \to X$ with $\gamma(a) = p$.
 \end{enumerate}
\end{lem}

Note that if $X$ is approximating, we can always join the future- and past-directed curves from point (iii) to find a timelike curve $\gamma \colon (a,b) \to X$ through $p$.

\begin{proof}
 (i) $\implies$ (ii) Let $p \in X$ be any point. Then $p \in J^+(p)$, so if $X$ is future-approximating, we get that $\emptyset \neq J^+(p) \subseteq \overline{I^+(p)}$. This implies that $I^+(p) \neq \emptyset$.
 
 (ii) $\implies$ (iii) By assumption, there exists points $q \in I^+(p)$. By causal path-connectedness, there exists a future-directed timelike curve $\gamma$ from $p$ to $q$, which by Definition~\ref{defcauscurv} must be non-constant, even if $p=q$. We can then remove the appropriate endpoint of $\gamma$ to get the desired curve.
 
 (iii) $\implies$ (i) Let $p \in X$ and $\gamma \colon [a,b) \to X$ be a future-directed timelike curve with $\gamma(a) = p$. By continuity of $\gamma$, we have $p = \lim_{s \to a} \gamma(s)$, so $p \in \overline{I^+(p)}$.
 
 The past statements are proved analogously.
\end{proof}

We can now easily see how Lorentzian length spaces are particular cases of the more general pre-length spaces that we will be working with in the rest of the paper.

\begin{prop} \label{LLSareapprox}
 If $X$ is a localizable, causally path-connected Lorentzian pre-length space, then $X$ is approximating and $d$-compatible. If $X$ is a Lorentzian length space, then $X$ is an approximating Lorentzian pre-length space with limit curves.
\end{prop}

\begin{proof}
 If $X$ is localizable, then by property (i) in Definition~\ref{defloc}, $X$ is $d$-compatible. Furthermore, by property (ii), every point $q \in X$ has $I^\pm(q) \neq \emptyset$, and then by Lemma \ref{pathapprox}, $X$ is approximating. The second statement follows trivially from the definitions.
\end{proof}

\begin{rem}
In connection with the null distance on Lorentzian pre-length spaces, Kunzinger and Steinbauer~\cite[Definition 3.4]{KuSt} introduced the notion of \emph{sufficiently causally connectedness (scc)}. A Lorentzian pre-length space is scc if it is path-connected (in the sense of metric spaces), causally path-connected (Definition~\ref{path-connected}) and every point $p \in X$ lies on some timelike curve $\gamma$. While the last condition is reminiscent of property (iii) in our Lemma~\ref{pathapprox}, it is in fact weaker, since scc puts no restriction on whether $p$ should be a future (or past) endpoint of $\gamma$. On the other hand, we do not need to assume path-connectedness.
\end{rem}

Having established the existence of causal (even timelike) curves through every point in Lemma~\ref{pathapprox}, the question remains whether one can find a (doubly-)inextendible causal curve through every point (see Definition~\ref{definex}). The following proposition and corollary answer this question in the affirmative, which will be crucial in Section~\ref{ghsec} when studying Cauchy sets. We need to assume that $(X,d)$ is proper in order to invoke the limit curve theorem. The use of the latter is also the reason why we only prove the existence of intextendible causal (and not timelike) curves.

\begin{prop}[Existence of maximal extensions of causal curves] \label{exisinex}
 Let $X$ be an approximating Lorentzian pre-length space with limit curves. Suppose, in addition, that $(X,d)$ is proper. Then, for every future-(past-)directed causal curve $\gamma \colon [a,b) \to X$ with $b < \infty$, there exists $c \in [b,\infty]$ and a future-(past-)inextendible causal curve $\lambda : [a,c) \to X$ such that $\lambda \vert_{[a,b)} = \gamma$.
\end{prop}

\begin{proof}
 Consider, without loss of generality, the case that $\gamma$ is future-directed. If $\gamma$ is already inextendible, there is nothing to prove since we can just choose $c=b$. Hence we consider the case of $\gamma$ being extendible. Then $\gamma$ has an endpoint, which we will, by abuse of notation, denote as $\gamma(b)$. Since $X$ is approximating, by Lemma~\ref{pathapprox} there is a future-directed timelike curve starting at $\gamma(b)$. Concatenating it with $\gamma$, we get a proper extension $\tilde \gamma \colon [a,c) \to X$ of $\gamma$, where $c > b$. If we can choose $\tilde \gamma$ to be inextendible, we are done. Hence, suppose for the sake of contradiction that all extensions of $\gamma$ are themselves extendible. There are two possible cases:
 \begin{enumerate}
  \item[1.] There exists a constant $C>0$ such that the $d$-arclength of all extensions of $\gamma$ is bounded by $C$. Suppose that we have chosen $C$ as small as possible. Then there exists a sequence $(\gamma_n)_n$ of extensions such that $L^d(\gamma_n) \to C$. Since all the $\gamma_n$ are extendible (hence we can add their future-endpoints) and agree at the point $\gamma(b)$, by Theorem~\ref{thm37} a subsequence converges to a limit curve $\gamma_\infty \colon [a,c] \to X$ of arclength $L^d(\gamma_\infty) = C$. But then, by the above, $\gamma_\infty$ admits a future extension, which is then also an extension of $\gamma$ and has arclength greater than $C$, a contradiction.
  \item[2.] There exists a sequence $(\gamma_n)_n$ of extensions of $\gamma$ such that $L^d(\gamma_n) \to \infty$. In this case we can apply Theorem \ref{thm314} to find an inextendible limit curve $\gamma_\infty$ of a subsequence. This $\gamma_\infty$ is then the desired inextendible extension of $\gamma$. \qedhere
 \end{enumerate}
\end{proof}

Combining Lemma~\ref{pathapprox} and Proposition~\ref{exisinex} gives us an important conclusion.

\begin{cor} \label{decomp1}
 Let $X$ satisfy the assumptions of Proposition~\ref{exisinex}. Then, for every point $p \in X$, there exists a doubly-inextendible causal curve passing through $p$. \hfill\qed
\end{cor}

\subsection{Causality conditions and time functions}\label{ssec:causalitytime}

The conditions in the previous subsections relating the topology and the causal structure (such as approximating) are satisfied automatically when the topology is that of a manifold, and the causal structure is induced by a Lorentzian metric. They can thus be thought of as making our Lorentzian pre-length spaces more ``manifold-like'', while still being much more general. In this subsection, on the other hand, we are going to discuss causality conditions, i.e., steps on the causal ladder, which are not satisfied by all smooth spacetimes and hence also not by all Lorentzian pre-length spaces. They should be thought of as criteria for physical reasonability.

In this section, we consider the notions of causality and non-total imprisonment, and the definition of time functions (for an in-depth treatment of the causal ladder for Lorentzian length spaces, see \cite{ACS}). Most of the material is standard, but Theorem~\ref{cor315} and Proposition~\ref{cau-nti} are new. The goal of this paper is to characterize the existence of (certain kinds of) time functions by suitable causality conditions, which will be introduced in the main sections. The causality conditions in this section are weaker, but also play an important role.

\medskip
A smooth spacetime is called causal if it contains no closed causal curves. The following equivalent definition is better suited for Lorentzian pre-length spaces.

\begin{defn}[{\cite[Definition 2.35]{KuSa}}]
 A Lorentzian pre-length space is called \emph{causal} if for any two points $p,q \in X$, $p < q$ implies $q \not< p$.
\end{defn}

Time functions too, can be defined either via causal curves (time functions are then required to be strictly increasing on future-directed causal curves), or in the following, more order-theoretic manner.

\begin{defn}
 A function $f \colon X \to \real$ on a Lorentzian pre-length space $(X,d,\tll,\leq,\tau)$ is called a \emph{generalized time function} if for all $p,q \in X$,
 \[
  p < q \implies f(p) < f(q).
 \]
 It is called a \emph{time function} if it is also continuous.
\end{defn}

Clearly, the existence of a (generalized) time function requires that the underlying space is at least causal.

\medskip
In the smooth case, non-total imprisonment is equivalent to the following definition (see, for instance, \cite[Theorem 4.39]{Min3}).

\begin{defn}[{\cite[Definition 2.35]{KuSa}}]
 A Lorentzian pre-length space $(X,d,\tll,\allowbreak \leq,\tau)$ is called \emph{non-totally imprisoning} if for every compact set $K \subseteq X$ there exists a constant $C>0$ such that for every causal curve $\gamma$ with image in $K$, $L^d(\gamma) \leq C$.
\end{defn}

As a corollary to the limit curve theorems, we obtain the following alternative characterizations.

\begin{thm} \label{cor315}
 Let $X$ be a Lorentzian pre-length space with limit curves. Then the following are equivalent.
 \begin{enumerate}
  \item $X$ is non-totally imprisoning.
  \item No compact set in $X$ contains a future-inextendible causal curve.
  \item No compact set in $X$ contains a past-inextendible causal curve.
  \item No compact set in $X$ contains a doubly-inextendible causal curve.
 \end{enumerate}
\end{thm}

\begin{proof}
 The equivalence between (i), (ii) and (iii) is shown in \cite[Corollary 3.15]{KuSa}. As a consequence of Lemma \ref{lem312}, any doubly-inextendible curve has infinite arclength. Thus (i) implies (iv).
 
 It remains to be shown that (iv) implies (i). Suppose $X$ is not non-totally imprisoning. Then there exists a compact set $K$ and a sequence of future-directed causal curves $\gamma_n \colon [0,L_n] \to X$, parametrized by arclength and contained in $K$, such that $L_n = L^d(\gamma_n) \to \infty$. Consider the sequence of future-directed causal curves $\bar\gamma_n \colon [0,L_n/2] \to X$ given by
 \[
  \bar\gamma_n(s) : = \gamma_n\left(\frac{L_n}{2} + s\right).
 \]
 Then also $L^d(\bar\gamma_n) = L_n/2 \to \infty$ and we can apply Theorem \ref{thm314} to find a converging subsequence $(\bar\gamma_{n_k})_k$ and a future-inextendible causal limit curve $\bar\gamma\colon [0,\infty) \to X$.  Next consider the sequence of past-directed causal curves $\tilde\gamma_k \colon [-L_n/2,0] \to X$ given by
 \[
  \tilde\gamma_k(s) := \gamma_{n_k}\left(\frac{L_n}{2}+s\right).
 \]
 Again we can apply Theorem~\ref{thm314} to find a converging subsequence $(\tilde\gamma_{k_m})_m$ and a past-inextendible limit curve $\tilde\gamma \colon (-\infty,0] \to X$. Note that
 \[
  \tilde\gamma(0) = \lim_{m \to \infty} \gamma_{n_{k_m}}(L_n/2) = \bar\gamma(0),
 \]
where the limit in the middle exists by compactness of $K$, a fact that we had already used implicitly when applying the Limit Curve Theorem~\ref{thm314}. Thus the curve $\tilde\gamma$ joined with $\bar\gamma$ is a doubly-inextendible causal curve $(-\infty,\infty) \to X$ contained in $K$. This contradicts our assumption (iv).
\end{proof}

The following result was shown by Kunzinger and S\"amann for Lorentzian length spaces, but only the assumption of causal path-connectedness is used in the proof.

\begin{prop}[{\cite[Theorem 3.26]{KuSa}}] \label{cau-nti}
 Suppose $(X,d,\tll,\leq,\tau)$ is a causally path-connected Lorentzian pre-length space. If $X$ is non-totally imprisoning, then $X$ is causal.
\end{prop}

For spacetimes, it is well-known that the existence of a time function implies strong causality, and that strong causality implies non-total imprisonment. This result has been shown by Kunzinger and Steinbauer for Lorentzian length spaces \cite[Theorem 3.13]{KuSt}, under the additional assumption that the time function must be topologically locally anti-Lipschitz. We instead give a direct proof of the fact that for Lorentzian pre-length spaces with limit curves, the existence of any kind of time function implies non-total imprisonment.

\begin{lem} \label{lem2}
 Let $(X,d,\tll,\leq,\tau)$ be a Lorentzian pre-length space with limit curves. If $X$ admits a time function, then $X$ is non-totally imprisoning.
\end{lem}

\begin{proof}
 Suppose $X$ is not non-totally imprisoning. Then by Theorem \ref{cor315} there exists a compact set $K \subseteq X$ and a future-inextendible future-directed causal curve $\gamma \colon [0, \infty) \to K$. Note the following two facts:
 \begin{enumerate}
  \item By Lemma \ref{lem312}, because $\gamma$ is inextendible, $\lim_{s \to \infty} \gamma(s)$ does not exist.
  \item By compactness of $K$, for every sequence $(s_i)_i$ in $[0,\infty)$, there exists a subsequence of $\left(\gamma(s_i)\right)_i$ that converges in $K$. 
 \end{enumerate}
 Thus we can find two sequences $(r_i)_i$ and $(s_i)_i$ in $[0,\infty)$ such that $p := \lim_{i \to \infty} \gamma(r_i) \\ \neq \lim_{i \to \infty} \gamma(s_i) =: q$ (in particular, both limits exist).
 
 For these $p, q$, pick $\delta > 0$ small enough so that $B_\delta (p)$ is contained in a weakly causally closed neighborhood and $q \not\in B_\delta (p)$. Assume w.l.o.g.\ that $r_1 < s_1 < r_2 < s_2 \ldots$ and that for all $i \in \nat$ we have $\gamma(r_i) \in B_{\delta/2} (p)$ and $\gamma(s_i) \not\in B_{\delta/2}(p)$. Now define a third sequence $(a_i)_i$ with $r_i < a_i < s_i$ and such that $a_i$ is the value at which $\gamma \vert_{[r_i, s_i]}$ first intersects $\partial B_{\delta/2}(p)$. We then have $r_i < a_i < s_i < r_{i+1}$. By compactness of $\partial B_{\delta/2}(p)$, there exists a subsequence of $(\gamma(a_i))_i$ that converges to a point $q' \neq p$.
 Since $r_i < a_i$, we have $\gamma(r_i) \leq \gamma(a_i)$. Because $B_\delta (p)$ is contained in a weakly causally closed neighborhood, we have $p < q'$. By assumption, $X$ admits a time function $t\colon X \to \real$, for which it holds that $t(p) < t(q')$. On the other hand, since $a_i < r_{i+1}$, we have $\gamma(a_i) \leq \gamma(r_{i+1})$. Thus $t(\gamma(a_i)) \leq t(\gamma(r_{i+1}))$ for all $i \in \nat$ and by continuity of $t$ and $\gamma$ also $t(q') \leq t(p)$. Combining this with the previous inequality, we obtain
 \[
  t(p) < t(q') \leq t(p),
 \]
 which is a contradiction.
\end{proof}

%%%%%%%%%%%%%%%%%%%%%%%%%%%%%%%%%%%%%%%%%% K-CAUSALITY %%%%%%%%%%%%%%%%%%%%%%%%%%%%%%%%%%%%%%

\section{Time functions and {$K$}-causality} \label{Kcausality}

The notion of $K$-causality was first introduced by Sorkin and Woolgar~\cite{SoWo} to study spacetimes with continuous Lorentzian metrics. Among the multiple applications of this concept, we emphasize the work of Minguzzi~\cite{Min2}, who showed that for smooth spacetimes, $K$-causality is equivalent to stable causality. Since Hawking~\cite{Haw} had shown earlier that stable causality is equivalent to the existence of a time function, so is $K$-causality. Minguzzi in \cite{Min2} also gave a direct proof of the equivalence between $K$-causality and the existence of time functions, which is more mathematically rigorous and less dependent on the Lorentzian manifold structure. Since $K$-causality is a purely order-theoretical notion, it can be used verbatim\footnote{In \cite{SoWo} the $K^+$-relation is only required to contain $I^+$.  Our definition with $J^+$ can be traced back to \cite{ACS, Min2}. On spacetimes and approximating Lorentzian pre-length spaces, we have $J^+ \subseteq \overline{I^+}$, hence there it makes no difference.} for Lorentzian pre-length spaces.

\begin{defn}[{\cite[Definitions 8 and 9]{SoWo}}] \label{Kdef}
 Let $(X,d,\tll,\leq,\tau)$ be a Lorentzian pre-length space. The \emph{$K^+$-relation} on $X$ is defined as the (unique) smallest transitive relation that contains $\leq$ and is (topologically) closed. 
 
A Lorentzian pre-length space $X$ is called \emph{$K$-causal}\footnote{In \cite{ACS,KuSa} a Lorentzian pre-length space with this property is called \emph{stably causal} because in the smooth case $K$-causality and stable causality are equivalent \cite{Min}. The definition of stable causality \cite[p.\ 433]{Haw}, however, requires knowledge about causal properties of ``nearby'' Lorentzian metrics. Since (causal) stability of Lorentzian pre-length spaces has not yet been investigated in this sense, we prefer to use the standard term $K$-causal.} if the $K^+$-relation is antisymmetric.
\end{defn}

In this section we establish the equivalence of the existence of time functions and $K$-causality on certain Lorentzian pre-length spaces (see Theorem~\ref{Kthm} below and Theorem~\ref{thm1} formulated for Lorentzian length spaces). This result generalizes the analogous theorem known for smooth spacetimes by Minguzzi~\cite[Theorem 7]{Min2} and the proof is obtained along the same lines.

\begin{thm} \label{Kthm}
 Suppose $\tau$ is a second countable, locally compact Lorentzian pre-length space with limit curves. Then $X$ is $K$-causal if and only if $X$ admits a time function.
\end{thm}

Before proving the theorem, we show that time functions can exist more generally also on Lorentzian pre-length spaces that are not $K$-causal if they do not satisfy the limit curve property (Definition~\ref{llslimitcurves}).

\begin{ex}[There exist Lorentzian pre-length spaces that admit a time function but are neither strongly causal nor $K$-causal]\label{example1}

 Let $(X,d)$ be the Euclidean plane with coordinates $(t,x)$. For any pair of points $p_i = (t_i,x_i)$, $i=1,2$, let
 \begin{align*}
  p_1 \tll p_2 &:\iff t_1 < t_2, \\
  p_1 \leq p_2 &:\iff t_1 < t_2 \text{ or } p_1 = p_2,
 \end{align*}
 as depicted in Figure \ref{figex1}, and
 \[
  \tau(p_1,p_2) = t_2 - t_1.
 \]
 This equips $(X,d)$ with the structure of a Lorentzian pre-length space. Clearly, the $t$-coordinate is a time function. 

 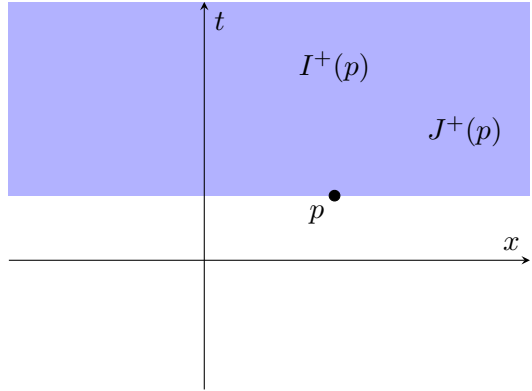
\begin{figure}
 \begin{center}
  \begin{tikzpicture}
   \begin{axis}[xmin=-1.5,xmax=2.5,ymin=-1,ymax=2,axis on top=true,
                 axis x line=middle,
                 axis y line=middle,
                 xlabel={$x$},
                 ylabel={$t$},
                 xtick=\empty,
                 ytick=\empty,
                 axis equal image]
     \path[name path=A] (axis cs: -1.5,0.5) -- (axis cs: 2.5,0.5);
     \path[name path=B] (axis cs: -1.5,2) -- (axis cs: 2.5,2);
     \addplot[blue!30] fill between [of=A and B];
     \node at (axis cs: 1,1.5) {$I^+(p)$};
     \node at (axis cs: 2,1) {$J^+(p)$};
     \filldraw[black] (axis cs: 1,0.5) circle (2pt) node[anchor=north east] {$p$};
    \end{axis}
   \end{tikzpicture}
  \end{center}
  \caption{The sets $I^+(p)$ (blue, without boundary) and $J^+(p) = I^+(p) \cup \{p\}$ for a point $p$ in Example \ref{example1}.} \label{figex1}
\end{figure}
 
 However, this space is not $K$-causal, as any closed relation containing $\leq$ (or $\tll$) must contain the relation $\leq_R$ given by
 \[
  p_1 \leq_R p_2 \iff t_1 \leq t_2.
 \]
 But $\leq_R$ is not antisymmetric, so the $K$-relation will not be antisymmetric either, and our space is therefore not $K$-causal.

 Note that while $K$-causality implies strong causality on smooth spacetimes and locally compact Lorentzian length spaces \cite[Proposition 3.16]{ACS} this need not be the case for Lorentzian pre-length spaces. To see that $X$ is also not strongly causal, recall that a Lorentzian pre-length space $X$ is called \emph{strongly causal} if the Alexandrov topology, generated by
\[
 I(p,q) = \{r \in X \mid p \tll r \tll q \}, \qquad p,q\in X,
\]
agrees with the metric topology \cite[Definitions 2.4 and 2.35]{KuSa}.
 
The Alexandrov topology of the above example is generated by the open horizontal stripes, hence is strictly coarser than the Euclidean topology, and $X$ is therefore not strongly causal.
\end{ex}

\subsection{Proof of Theorem~\ref{Kthm}}

We follow Minguzzi's proof for smooth spacetimes \cite{Min2}. A key element is the following theorem from utility theory, a branch of mathematical economics with resemblances to causality theory.

\begin{thm}[Levin's Theorem {\cite{Lev}}] \label{Levin}
 Let $X$ be a second countable, locally compact Hausdorff topological space and $R$ be a closed preorder on $X$. Then there exists a continuous function $f: X \to \real$ such that
 \[
  (x,y) \in R \implies f(x) \leq f(y),
 \]
 with equality if and only if $x=y$.
\end{thm}

That $K$-causality implies the existence of a time function is a direct consequence of Levin's Theorem, and is in fact true in an even more general setting than ours, as already pointed out by Minguzzi \cite{Min5}.

It remains to show the converse. The next lemma gives us a more explicit characterization of the $K$-relation. Its proof is not significantly different to its smooth counterpart \cite[Lemma 3]{Min2}, but is included for the sake of clarity.
 
\begin{lem} \label{lem3}
 Let $(X,d,\tll,\leq,\tau)$ be a non-totally imprisoning, locally compact Lorentzian pre-length space with limit curves. If $(p,q) \in K^+\subseteq X \times X$, then either $p \leq q$ or for every relatively compact open set $B$ containing $p$, there exists $r \in \partial B$ such that $p < r$ and $(r,q) \in K^+$.
\end{lem}

\begin{proof}
 For the purposes of this proof, it will be more convenient to denote relations as subsets of $X \times X$; in particular, $J^+ := \{ (p,q) \in X \times X \mid p \leq q \}$. Consider the relation
 \begin{align*}
  R^+ := \{ (p,q) \in K^+ \mid & \, (p,q) \in J^+ \text{ or for every relatively compact } \\
  &\text{ open set } B  \text{ containing } p \text{ there is an } r \in \partial B \\
  &\text{ such that } p < r \text{ and } (r,q) \in K^+ \}.
 \end{align*}
 Clearly, $J^+ \subseteq R^+ \subseteq K^+$. We will show that $R^+$ is closed and transitive, which then implies $R^+=K^+$, and in turn proves the Lemma. Transitivity can be proven exactly in the same way as in the smooth case, so we refer to \cite[Lemma 3]{Min2}.
 
 To show closedness of $R^+$, consider $(p_n,q_n) \to (p,q)$ with $(p_n,q_n) \in R^+$ for all $n \in \nat$. If $p=q$ then $(p,q) \in J^+ \subseteq R^+$. We can therefore assume that $p \neq q$ and it remains to be shown that $(p,q) \in R^+$ as well. Let $B$ be an open relatively compact neighborhood of $p$. For sufficiently large $n$, $p_n \neq q_n$ and $p_n \in B$. By passing to a subsequence if necessary, we can assume that either $(p_n,q_n) \in J^+$ (case 1) or $(p_n, q_n) \in R^+ \setminus J^+$ (case 2) for all $n \in \nat$:
 \begin{enumerate}
  \item[1.] Suppose $(p_n,q_n) \in J^+$ for all $n \in \nat$. By assumption $X$ is causally path-connected (and $p_n \neq q_n$), thus there exist future-directed causal curves $\gamma_n \colon [0,1] \to X$ from $p_n$ to $q_n$. By passing to a subsequence if necessary, we can assume that $\gamma_n$ either lies entirely in $\overline{B}$ for all $n$, or leaves $\overline{B}$ for all $n$.
  
  If all the $\gamma_n$ lie inside $\overline{B}$, then by non-total imprisonment their lengths (and by linear reparametrization also their Lipschitz constants) are bounded above by a positive constant $C$ independent of $n$. Thus we are in conditions to apply Theorem \ref{thm37}, which shows the existence of a limit causal curve connecting $p$ and $q$, and hence $(p,q) \in J^+ \subseteq R^+$.
  
  If, on the other hand, none of the $\gamma_n$ lie entirely inside $\overline{B}$, then there is a first parameter value $s_n$ at which $\gamma_n$ leaves $B$ (and $\gamma_n(s_n) \in \partial B$ by connectedness). Define a new sequence of curves $\tilde \gamma_n := \gamma_n \vert_{[0,s_n]}$. Linear reparametrization so that $\tilde\gamma_n \colon[0,1]\to X$ together with Theorem~\ref{thm37} shows the existence of a limit causal curve that connects $p$ with a point $r = \lim_{n \to \infty} \gamma_n(s_n) \in \partial B$ (again w.l.o.g.\ by passing to a subsequence). Because $\gamma_n(s_n)\leq q_n$ for all $n$ it follows that $(r,q) \in \overline{J^+} \subseteq K^+$ by closedness of the $K$-relation. Since also $p < r$, we again conclude that $(p,q) \in R^+$.
  
  \item[2.] Suppose $(p_n,q_n) \in R^+ \setminus J^+$ for all $n$. Then there exist points $r_n \in \partial B$ such that $p_n < r_n$ and $(r_n,q_n) \in K^+$. By passing to a subsequence if necessary, we can assume that $r_n \to r \in \partial B$. Arguing as in case 1 (for the sequence $(p_n,r_n)$), either $p < r$, or there exists $r' \in \partial B$ such that $p < r'$ and $(r',r) \in \overline{J^+} \subseteq K^+$. Combining this with the fact that $(r,q) \in K^+$ by closedness of the $K$-relation, it follows that $(p,q) \in R^+$.
 \end{enumerate}
 In both cases we have thus shown that $(p,q) \in R^+$, which concludes the proof.
\end{proof}

The next and final lemma is key, and tells us that time functions are $K$-utilities, in the language of economics. In other words, a time function with respect to the causal relation $\leq$ is automatically also a time function with respect to the $K^+$-relation.

\begin{lem} \label{lem4}
 Let $(X,d,\tll,\leq,\tau)$ be a locally compact Lorentzian pre-length space with limit curves and $t\colon X \to \real$ be a time function. If $(p,q) \in K^+$, then either $p = q$ or $t(p) < t(q)$.
\end{lem}

\begin{proof}
 The proof is the same as the one of \cite[Lemma 4]{Min2}, replacing \cite[Lemma 2]{Min2} and \cite[Lemma 3]{Min2} by our Lemmas \ref{lem2} and \ref{lem3}, respectively.
\end{proof}

\begin{proof}[Proof of Theorem \ref{Kthm}]
 That $K$-causal spaces admit time functions is a direct consequence of Levin's Theorem~\ref{Levin}. Conversely, if $X$ admits a a time function, then the $K^+$-relation must be antisymmetric, as otherwise it would contradict Lemma \ref{lem4}.
\end{proof}

\begin{rem}
 It is worth pointing out that, throughout this section, we have not made use of the chronological relation $\tll$, nor of timelike curves. Hence, Theorem \ref{Kthm} is still valid if in Definition~\ref{path-connected} we only require the existence of causal (and not of timelike) curves, or even if $\tll$ is empty.
\end{rem}

%%%%%%%%%%%%%%%%%%%%%%%%%%%%%%%%%%%%%%%%%%%%%%%%%%% VOLUME FUNCTIONS %%%%%%%%%%%%%%%%%%%%%%%%%%%%%%%%%

\section{Volume time functions} \label{volfcts}

In this section we introduce and explicitly construct special types of functions, called \emph{averaged volume functions}, on Lorentzian pre-length spaces that are equipped with probability measures. While the existence of a suitable measure solely depends on the topology (in fact, the metric structure) it is the causal structure that determines whether these functions are time functions. More precisely, we will see that the averaged volume functions are time functions if and only if the underlying Lorentzian pre-length space is causally continuous in Theorem~\ref{ccthm} (see Theorem~\ref{thm2} for the Lorentzian length space version thereof).

Our results generalize a classical theorem of Dieckmann~\cite{Die,Die2} (also stated earlier by Hawking and Sachs \cite{HaSa}, but with an incomplete proof). Volume functions had already been introduced earlier by Geroch \cite[Sec.\ 5]{Ger} to study global hyperbolicity; we will replicate those results in Section \ref{ghsec}. In this section, we follow the approach of Dieckmann, but also make use of an averaging procedure similar to that used by Hawking~\cite{Haw} to study stable causality and time functions. Besides that, our methods in this section are based on order- and measure-theoretical arguments, and we do not need to assume the existence of causal curves.

\subsection{Averaged volume functions}\label{ssec:avgvolfct}

To construct averaged volume functions on a Lorentzian pre-length space $X$, we equip $X$ with a Borel measure $\mu$ satisfying
\begin{enumerate}
 \item $\mu(X)=1$, i.e., $\mu$ is a probability measure, and
 \item $\supp(\mu) = X$, i.e., $\mu$ has full support.
\end{enumerate}
When $X$ is finite the construction of $\mu$ is trivial. Otherwise such a measure exists precisely when $(X,d)$ is a separable metric space (which is automatically satisfied for all compact spaces).

\begin{prop} \label{sep}
 A metric space $(X,d)$ admits a Borel probability measure $\mes$ with $\supp(\mes) = X$ if and only if it is separable (equivalently, second-countable).
\end{prop}

\begin{proof}
 Suppose $(X,d)$ is separable. Then it contains a countable dense subset $D= \{p_n \mid n \in \nat\}$. Denote by $\delta_n$ the Dirac delta measure centered at $p_n$, and define
 \[
\mu := \sum_{n \in \nat} 2^{-n} \delta_n.
 \]
 The measure $\mes$ has the desired properties since (i) $\mes(X) = \sum_n 2^{-n}=1$ and (ii)  for all open sets $A$, $A \cap D \neq \emptyset$ by denseness and hence $\mu(A) > 0$. The proof of the converse can be found in \cite[p.\ 134]{MaSi}.
 
 Finally, secound countability implies separability, and on metric spaces the two notions in fact are equivalent.
\end{proof}

Let $(X,d,\tll,\leq,\tau)$ be a Lorentzian pre-length space equipped with a Borel probability measure $\mu$ of full support. For $r \in (0,1)$ and $p \in X$, let
\begin{align*}
 \thck^\pm_r(p) &:= \left\{x \in X \mid d\left(x,I^\pm(p)\right) < r \right\}, \\
 \vol^\pm_r(p) &:= \mes\left(\thck^\pm_r(p)\right).
\end{align*}
We call $\thck^\pm_r(p)$ the $r$-thickening of $I^\pm(p)$, as depicted in Figure~\ref{figV}, and $\vol^\pm_r(p)$ its volume.

\begin{figure}%[ht]
 \begin{center}
  \begin{tikzpicture}[declare function={
    func(\x)= (\x<=-0.5) * (-\x-1)   +
     and(\x>-0.5, \x<=0.5) * (-sqrt(0.5-\x^2))     +
                (\x>0.5) * (\x-1);
  }]
   \begin{axis}[xmin=-2,xmax=2,ymin=-1,ymax=2,hide axis,axis equal image,y dir = reverse]
     \addplot[name path=B, domain=-2:2, fill=blue!40, draw=none]{abs(x)};
     \addplot[name path=A, domain=-2:2, samples=100, draw=none]{func(x)};
     \addplot[blue!20] fill between[of=A and B];
     \node at (axis cs: 0,1.25) {$I^-(p)$};
     \node at (axis cs: 1.4,0.85) {$r$};
     \node[style={rectangle, draw=black}] (thck) at (axis cs: -1.6,-0.3) {$\thck_r^-(p)$};
     \draw[->] (thck.south east) -- (axis cs: -0.75,0.25);
     \draw[->] (thck.south east) -- (axis cs: -0.8,1.4);
     \draw[<->] (axis cs:1, 1) -- (axis cs:1.5, 0.5);
     \filldraw[black] (axis cs: 0,0) circle (2pt) node[anchor=south east] {$p$};
    \end{axis}
   \end{tikzpicture}
  \end{center}
  \caption{The sets $I^-(p)$ (dark blue) and $\thck^-_r(p)$ (light and dark blue) for some point $p$ in Minkowski spacetime, with $d$ the Euclidean distance.}
  \label{figV}
\end{figure}
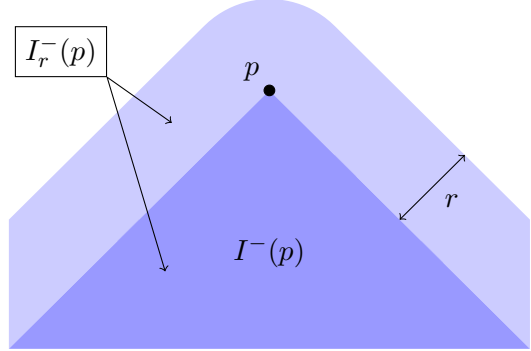

\begin{defn} \label{volfct}
 Let $(X,d,\tll,\leq,\tau)$ be a Lorentzian pre-length space equipped with a Borel probability measure $\mu$ of full support. The \emph{future} ($+$) and \emph{past} ($-$) \emph{averaged volume functions} of $\mu$ are defined by
 \[
  \tf^\pm(p) := \mp \int_0^1 \vol^\pm_r(p) \,dr, \qquad p \in X.
 \]
\end{defn}

Note that the integral exists for all points $p \in X$ because the function $r \mapsto \vol^\pm_r(p)$ is increasing and bounded.

\begin{rem}[Comparison to previous definitions of volume functions]\label{classvolfct}
The classical definition of a volume function by Geroch~\cite[Section 5]{Ger} is simpler and reads $\tf_\mathrm{cl}^\pm(p) = \mp \mes\left(I^\pm(p)\right)$. However, it was discovered by Dieckmann~\cite[Def.\ 1.2]{Die} that in order to show continuity, one has to require that $\mu$ also satisfies the property (iii) $\mu(\partial I^\pm (p)) = 0$. On a smooth spacetime $(M,g)$ one can always construct such an admissible measure $\mes$ from the volume form using a partition of unity and utilize that $\partial I^\pm(p)$ is a hypersurface having zero Lebesgue measure in charts \cite[Prop.\ 1.1]{Die}. Since we do not have a manifold structure and the Lebesgue measure at our disposal, we instead integrate over $r$ to ``average out'' discontinuities, hence the addition of ``averaged'' in the naming of volume functions in Definition~\ref{volfct}. This averaging procedure is inspired by the work of Hawking~\cite{Haw} on stable causality and time functions.
\end{rem}

\begin{rem}[$\mes$-dependence]
It is clear that the above constructions of $I^\pm_r$, $V^\pm_r$, and $t^\pm$ depend crucially on the choice of $d$ and $\mes$. We will, however, see that the \emph{existence\/} of (generalized) time functions is at this point independent of the particular choice of $\mes$ and also of $d$ (as long as the metric is second countable). More precisely, whether the averaged volume functions $t^\pm$ are indeed (continuous) time functions depends only on the causal structure of $X$.
\end{rem}

We end this subsection by proving that $t^\pm$ are \emph{isotone} or \emph{causal} functions (see \cite[Def.\ 1.17]{Min3}), a property that is weaker than being a time function.
 
\begin{lem} \label{nondec}
 Let $X$ be a Lorentzian pre-length space as in Definition~\ref{volfct} and let $p,q \in X$. If $p \leq q$, then $\vol^-_r(p) \leq \vol^-_r(q)$ and $\vol^+_r(p) \geq \vol^+_r(q)$. In particular,
 \[ p \leq q \Longrightarrow t^\pm(p) \leq t^\pm(q). \]
\end{lem}

\begin{proof}
 By the push-up Lemma~\ref{push-up}, $p \leq q$ implies $I^-(p) \subseteq I^-(q)$, and hence $\thck^-_r(p) \subseteq \thck^-_r(q)$ for all $r \in (0,1)$. The first conclusion thus follows from the monotonicity of $\mu$, and the second one from the monoticity of the integral.
\end{proof}

\subsection{Averaged volume functions as generalized time functions}

Finally, in order to show that $t^\pm$ are generalized time functions we need to apply the standard distinguishing causality condition, or at least our own weaker version thereof.

\begin{defn}[{\cite[p.\ 486]{KrPe}}] \label{disting}
 Let $(X,d,\tll,\leq,\tau)$ be a Lorentzian pre-length space. We say that
 \begin{enumerate}
  \item $X$ is \emph{past-distinguishing} if \[I^-(p) = I^-(q) \Longrightarrow p=q, \qquad  p,q \in X,\]
  \item $X$ is \emph{future-distinguishing} if \[I^+(p) = I^+(q) \Longrightarrow p=q, \qquad p,q \in X.\]
 \end{enumerate}
We call $X$ \emph{distinguishing} if it is both past- and future-distinguishing.
\end{defn}
 
\begin{defn}
We say that $X$ is \emph{causally} (\emph{past}- or \emph{future}-) \emph{distinguishing} if the conditions of Definition~\ref{disting} are only required to hold for all $p,q \in X$ with $p \leq q$.
\end{defn}

Furthermore, we assume that the Lorentzian pre-length spaces are approximating (see Section~\ref{ssec:approx}). This avoids the pathological situation where the future or past of a point could be ``far away'' from the point itself, or empty.

Having equipped our spaces with sufficient causal and topological conditions, we are in a position to establish the well-known classical result about generalized time functions \cite[Prop.\ 2.2]{Die} also for separable Lorentzian pre-length spaces (for which averaged volume functions $t^\pm$ from Definition~\ref{volfct} are well-defined by Proposition~\ref{sep}).

\begin{prop} \label{dist}
 Let $(X,d,\tll,\leq,\tau)$ be a past-(future-)approximating Lo\-rent\-zian pre-length space equipped with a Borel probability measure of full support. Then $X$ is causally past-(future-)distinguishing if and only if $\tf^-$ ($\tf^+$) is a generalized time function i.e., for all $p,q \in X$
 \[
  p < q \Longrightarrow t^\mp(p) < t^\mp (q).
 \]
\end{prop}

\begin{proof}
 We prove the past version. If $X$ is causally past-distinguishing, then for all $p,q \in X$ with $p<q$,
 \[
  I^-(p) \subsetneq I^-(q).
 \]
 Clearly, $q \not\in I^-(p)$, because $q \in I^-(p)$ would imply $I^-(q) \subseteq I^-(p) \subsetneq I^-(q)$, a contradiction. We show that also $q \not\in \partial I^-(p)$: For any $x \in I^{-}(q)$ the future $I^+(x)$ is open and thus contains an open set $V$ around $q$. If $q \in \partial I^-(p)$ then there is a point $y \in V \cap I^-(p)$, thus $x \tll y \tll p$ and by transitivity $x \in I^-(p)$. Therefore $I^-(q) \subseteq I^-(p)$, a contradiction.

 Since any metric space is normal, the disjoint closed sets $\{q\}$ and $\overline{I^-(p)}$ can be separated by disjoint open sets. In particular, there exists an open neighborhood $U$ around $q$ such that $r_0 := d(U,I^-(p)) >0$. The intersection $U \cap I^-(q)$ is open, and by the past-approximating property of $X$ at $q$ (see Lemma~\ref{lemapprox}) also non-empty. Thus, since $U \cap I^-(q) \subseteq I_r^-(q) \setminus I_r^-(p)$ for all $r<r_0$,
 \begin{align*}
  \tf^-(q)-\tf^-(p) &= \int_0^1 \mes\left(\thck_r^-(q) \setminus \thck_r^-(p)\right) dr \\
  &\geq \int_0^{\min(r_0,1)} \mes\left(U \cap I^-(q)\right) dr > 0.
 \end{align*}

 To prove the converse, assume that for all $p<q$ we have $\tf^-(q)-\tf^-(p) > 0$. In order to show that $X$ is causally past-distinguishing furthermore assume that $I^-(p) = I^-(q)$ and $p \leq q$. Then $I^-_r(p) = I^-_r(q)$ for all $r \in [0,1]$, and therefore $t^-(p)=t^-(q)$. Hence $p=q$.
\end{proof}

Example~\ref{example1} is a Lorentzian pre-length space that is causally distinguishing but not distinguishing (any two points on the same level set of $t$ have the same past, but are not causally related to one another). One can show that under certain conditions the two notions agree, in particular, for smooth spacetimes.

\begin{prop}\label{distingprop}
 Let $(X,d,\tll,\leq,\tau)$ be a locally compact, causally path-connected, locally weakly causally closed, (past-/future-)approximating Lorent\-zian pre-length space. If $X$ is causally (past-/future-)distinguishing, then it is also (past-/future-)distinguishing.
\end{prop}

\begin{proof}
 Suppose $X$ is causally past-distinguishing but not past-distinguishing (the future case is analogous). Then there exist two distinct points $p,q \in X$ such that $I^-(p) = I^-(q)$. By Lemma~\ref{lemapprox} there is a sequence $(p_n^-)_n$ that approximates $p$ from the past. Because $p_n^- \in I^-(p) = I^-(q)$ and $X$ is causally path-connected, there exists a sequence of causal curves $(\gamma_n)_n$ connecting $p_n^-$ and $q$. Let $\delta$ be small enough so that $\ball_\delta(p)$ is compact and does not contain $q$. Consider the sequence of points $(r_n)_n$ where $\gamma_n$ first intersects $\partial B_{\delta/2}(p)$. By compactness, we may assume it converges to a point $r \in \partial B_{\delta/2}(p)$. Without loss of generality suppose that $\ball_\delta(p)$ is contained in a weakly causally closed neighborhood, and so furthermore $p \leq r$ and $I^-(p) \subseteq I^-(r)$ by the push-up Lemma~\ref{push-up}. Moreover, by openness of $\tll$ we have $I^-(r) \subseteq \bigcup_n I^-(r_n)$, and since $r_n \tll q$, by transitivity we have $I^-(r_n)\subseteq I^-(q)$. Hence $I^-(p) \subseteq I^-(r) \subseteq I^-(q)$, which together with our initial assumption $I^-(p)=I^-(q)$ implies that $I^-(p) = I^-(r)$. Since $X$ is causally past-distinguishing we thus know that $p=r$, which contradicts the fact that $d(p,r)=\frac{\delta}{2}>0$. Thus $X$ must indeed also be past-distinguishing.
\end{proof}

\subsection{Continuity of averaged volume functions}

We conclude this section by generalizing the equivalence between causal continuity and the continuity of volume functions (shown in the smooth case by Dieckmann~\cite[Proposition 2.5]{Die}) to establish when $t^\pm$ of Definition~\ref{volfct} are time functions.

\begin{defn}
 Let $(X,d,\tll,\leq,\tau)$ be a Lorentzian pre-length space. If for all $p,q \in X$,
 \begin{enumerate}
  \item $I^+(p) \subseteq I^+(q)$ implies $I^-(q) \subseteq I^-(p)$, we say that $X$ is \emph{past reflecting}.
  \item $I^-(p) \subseteq I^-(q)$ implies $I^+(q) \subseteq I^+(p)$, we say that $X$ is \emph{future reflecting}.
 \end{enumerate}
 If $X$ is both past and future reflecting, we say that it is \emph{reflecting}.
\end{defn}

We define causal continuity as in the smooth case. Ak\'e et al.\ showed that this notion of causal continuity implies strong causality \cite[Proposition 3.15]{ACS}, as it does in the smooth case. However, it turns out that the optimal condition for our Theorem \ref{ccthm} is a weaker version thereof. This weaker version is not (trivially) sufficient for the results of Ak\'e et al.\ to still hold. In any case, both definitions are equivalent when the conditions of Proposition~\ref{distingprop} are met (in particular, in the smooth case).

\begin{defn}[{\cite[Definition 3.9]{ACS}}]
 A Lorentzian pre-length space is called \emph{causally continuous} if it is reflecting and distinguishing.
\end{defn}

\begin{defn}
 A Lorentzian pre-length space is called \emph{weakly causally continuous} if it is reflecting and causally distinguishing.
\end{defn}

Using this terminology we can state the main result of this section (see Theorem~\ref{thm2} for the corresponding Lorentzian length space version).

\begin{thm}\label{ccthm}
 Let $X$ be an approximating Lorentzian pre-length space equipped with a Borel probability measure of full support. Then $X$ is weakly causally continuous if and only if the averaged volume functions $t^\pm$ are time functions.
\end{thm}

Note that also in this section we combine a topological condition (second countable) and a causal condition ((weakly) causally continuous) to characterize the existence of volume time functions.

\begin{proof}
 By Proposition~\ref{dist} the averaged volume functions are generalized time functions if and only if $X$ is causally distinguishing. In Lemma~\ref{refl} below we show that their continuity is characterized by $X$ being reflecting.
\end{proof}
 
First we show that property (iii) discussed in Remark~\ref{classvolfct} holds for almost all thickenings of $\thck^\pm(p)$.
 
\begin{lem} \label{sameae}
 Let $(X,d,\tll,\leq,\tau)$ be a Lorentzian pre-length space equipped with a Borel probability measure $\mes$ of full support. For $p \in X$ consider the functions $r \mapsto \vol^\pm_r(p) = \mu(\thck^\pm_r(p))$ as defined in Section~\ref{ssec:avgvolfct} and
 \[
  r \mapsto \overline{\vol^\pm_r}(p) := \mes\left(\overline{\thck^\pm_r(p)}\right).
 \]
 Then the sets
 \[
  S^\pm(p) := \left\{ r \in (0,1) \mid \overline{\vol^\pm_r}(p) \neq \vol^\pm_r(p) \right\}
 \]
 are countable in $(0,1)$ for all $p \in X$.
\end{lem}

\begin{proof}
 Let $p \in X$. By definition and additivity of $\mes$ we have
 \begin{align*}
  \overline{\vol^\pm_r}(p) - \vol^\pm_r(p) &= \mes(\overline{\thck^\pm_r(p)})-\mes(\thck^\pm_r(p)) \\ &= \mes \left( \overline{\thck^\pm_r(p)} \setminus \thck^\pm_r(p) \right) = \mes(\partial \thck^\pm_r(p)).
 \end{align*}
 Hence we can rewrite $S^\pm(p)$ as
 \[
  S^\pm(p) =\left\{ r \in (0,1) \ \middle\vert \ \mes \left(\partial \thck^\pm_r(p)\right)  > 0 \right\}.
 \]
 To see that $S^\pm(p)$ is countable, consider the sets
 \[
  S_n^\pm(p) :=\left\{ r \in (0,1) \ \middle\vert \ \mes \left( \partial \thck^\pm_r(p) \right) > \frac{1}{n} \right\},
 \]
 for $n \in \nat$. We show that $|S_n^\pm(p)| \leq n$. Otherwise there would exist at least $n+1$ distinct $r_i \in S_n^\pm(p)$. Since all $\partial\thck^\pm_{r_i}(p)$ are disjoint, this would imply that
 \[
  \frac{n+1}{n} < \sum_{i=1}^{n+1} \mes(\partial\thck^\pm_{r_i}(p)) \leq \mes(X) = 1,
 \]
 a contradiction. Finally, because $S^\pm(p) = \bigcup_{n\in\nat} S_n^\pm(p)$, we deduce that the sets $S^\pm(p)$ are countable for any $p \in X$.
\end{proof}

In addition to Lemma~\ref{sameae}, the continuity of the averaged volume functions rests on the following general result which is based on Lebesgue's dominated convergence theorem (in \cite{Jost} formulated for $X=\real^d$ but true for all sequential spaces).

\begin{thm}[{\cite[Theorem 16.10]{Jost}}]\label{jostthm}
 Let $(X,d)$ be a metric space, $U \subseteq X$ and $y_0 \in U$. Consider a function $f \colon \real^n \times U \to \real \cup \{ \pm \infty\}$. Assume that
 \begin{enumerate}
   \item for every fixed $y \in U$ the function $x \mapsto f(x,y)$ is integrable,
   \item for almost all $x \in \real^n$ the function $y \mapsto f(x,y)$ is continuous at $y_0$, %and
   \item there exists an integrable function $F \colon \real^n \to \real \cup \{+\infty\}$ with the property that for every $y \in U$,
   \[
    |f(x,y)| \leq F(x)
   \]
   holds almost everywhere on $\real^n$.
  \end{enumerate}
  Then the function
  \[
   g(y) := \int_{\real^n} f(x,y) \, dx
  \]
  is continuous at the point $y_0$.
 \end{thm}

With this, we can prove the last lemma of this section (compare with \cite[Proposition 1.6]{Die}). Together with Proposition \ref{dist}, it constitutes the proof of Theorem \ref{ccthm}.

\begin{lem} \label{refl}
Let $(X,d,\tll,\leq,\tau)$ be an approximating Lorentzian pre-length space equipped with a Borel probability measure $\mes$ of full support. For any point $q \in X$, the following are equivalent:
 \begin{enumerate}
  \item $\tf^-$ ($\tf^+$) is continuous at $q$.
  \item $I^+(q) \subseteq I^+(p) \implies I(\tilde{p},p) \cap I^-(q) \neq \emptyset$ for all $\tilde{p} \tll p$ \\ ($I^-(q) \subseteq I^-(p) \implies I(p,\tilde{p}) \cap I^-(q) \neq \emptyset$ for all $p \tll \tilde{p}$).
  \item $X$ is past (future) reflecting at $q$.
  \item \ $\bigcap_{q \tll x} I^-(x) \subseteq \overline{I^-(q)}$ \\ ($\bigcap_{x \tll q} I^+(x) \subseteq \overline{I^+(q)}$).
 \end{enumerate}
\end{lem}

\begin{proof}
 We show the past versions. Fix $q \in X$.
 
(i) $\Longrightarrow$ (ii) Assume (i) holds but not (ii). Then there exist points $p$ and $\tilde{p}$ such that $\tilde{p} \tll p$ and $I^+(q) \subseteq I^+(p)$ but $I(\tilde{p},p) \cap I^-(q) = \emptyset$. Note that $I(\tilde p , p)$ is open and non-empty, because we can approximate $p$ from the past, and by openness of $\tll$, any past-approximating sequence must enter $I^+(\tilde{p})$. Let $r_0 > 0$ be small enough, then also the set
 \[
  B := \left\{ x \in I(\tilde p , p) \mid d(x, I^-(q)) > r_0 \right\} = I(\tilde p,p) \setminus \overline{I^-_{r_0}(q)}.
 \]
is open and non-empty: If not, then $I(\tilde p,p) \subseteq \overline{I^-_{r_0}(q)}$. But since $I(\tilde p,p)$ is open, this would mean that $I(\tilde{p},p) \cap I^-(q) \neq \emptyset$, contradicting our earlier assumption. Clearly, $B \cap I_{r}^-(q) = \emptyset$ for all $r < r_0$. 

Let $(q_n^+)$ be a sequence approximating $q$ from the future. Since $q_n^+ \in I^+(q) \subseteq I^+(p)$, transitivity of $\tll$ implies that $I(\tilde{p},p) \subseteq I^-(p) \subseteq I^-(q_n^+)$. Then, because
 \[
  \tf^-(q_n^+) - \tf^-(q) = \int_0^1 \mes\left(\thck^-_r(q_n^+) \setminus \thck^-_r(q)\right) dr
 \]
and $B \subseteq \thck^-_r(q_n^+) \setminus \thck^-_r(q)$ for $r<r_0$, it follows that
 \[
  \tf^-(q_n^+) - \tf^-(q) \geq \int_0^{r_0} \mes(B) \, dr > 0,
 \]
 showing discontinuity of $\tf^-$ at $q$, a contradiction to (i).
 
 (ii) $\Longrightarrow$ (iii) Assume (ii) holds but not (iii). If $X$ is not past reflecting at $q$, then there exists a $p \in X$ with $I^+(q) \subseteq I^+(p)$ but $I^-(p) \not\subseteq I^-(q)$. Thus there is a $\tilde{p} \in I^-(p) \setminus I^-(q)$ and by (ii) there exists a point $\hat{p} \in I(\tilde{p},p) \cap I^-(q)$. But then $\tilde{p} \tll \hat{p} \tll q$, and by transitivity $\tilde{p} \tll q$, a contradiction.
 
 (iii) $\Longrightarrow$ (iv) If $\bigcap_{q\tll x} I^-(x) = \emptyset$ the conclusion is trivial. Suppose $p \in \bigcap_{q \tll x} I^-(x)$. Then $x \in I^+(p)$ for all $x \gg q$, i.e., $I^+(q) \subseteq I^+(p)$. By (iii) $X$ is past reflecting at $q$, hence $I^-(p) \subseteq I^-(q)$. Since $X$ is past-approximating, $p \in J^-(p) \subseteq \overline{I^-(p)} \subseteq \overline{I^-(q)}$. Since $p$ was an arbitrary point in the intersection, (iv) follows.
 
 (iv) $\Longrightarrow$ (i) Fix $r \in (0,1)$ and let $\epsilon > 0$. Because $X$ is past-approximating, we can find a sequence $(q^-_i)_i$ that approximates $q$ from the past. By openness of $\tll$, it then follows that
 \[
  \bigcup_{i=1}^\infty \thck^-_r (q^-_i) = \thck^-_r (q),
 \]
 and by standard measure theory \cite[Theorem 1.2.5]{KrPa} there exists $i_0 \in \nat$, such that
 \[
  \mes\left(\thck^-_r (q)\right) - \mes\left(\thck^-_r (q^-_{i_0})\right) < \epsilon.
 \]
 By Lemma \ref{nondec}, we deduce
 \begin{equation} \label{past}
  \vol^-_r(q) - \vol^-_r(p) < \epsilon \quad \text{ for all } \quad p \in I^+(q^-_{i_0}).
 \end{equation}
 Next, consider $(q^+_j)_j$ a sequence approximating $q$ from the future. Assumption (iv) implies that
 \[
  \bigcap_{j=1}^\infty \thck_r^-(q^+_j) \subseteq \overline{\thck_r^-(q)},
 \]
 and hence
 \[
  \mes\left( \bigcap_{j=1}^\infty \thck_r^-(q^+_j) \right) \leq \mes\left( \overline{\thck_r^-(q)}\right).
 \]
 By \cite[Theorem 1.2.5]{KrPa}, for $r$ given, there exists $j_0 \in \nat$ such that
 \[
  \mes\left( \thck_r^-(q^+_{j_0}) \right) - \mes\left( \overline{\thck_r^-(q)}\right) < \epsilon.
 \]
 Then, using Lemma \ref{nondec} we deduce
 \begin{equation} \label{future}
  \vol^-_r(p) - \overline{\vol^-_r}(q) < \epsilon \quad \text{ for all } \quad p \in I^-(q^+_{j_0}).
 \end{equation}
 By Lemma \ref{sameae}, $\overline{\vol^-_r}(q) = {\vol^-_r}(q)$ for all but countably many $r$. Hence for almost all $r\in (0,1)$, we can combine \eqref{past} and \eqref{future} to write
 \[
  \vert \vol_r^-(q) - \vol_r^-(p) \vert < \epsilon \quad \text{ for all } \quad p \in I(q^-_{i_0},q^+_{j_0}).
 \]
 We conclude that almost all functions $V^-_r \colon X \to \real$, $r \in (0,1)$, are continuous at $q$. Thus it follows from Theorem~\ref{jostthm} that%{a standard result in integration theory \cite[Theorem 16.10]{Jost}}%\footnote{We need to extend \cite[Theorem 16.10]{Jost} from $\real^d$ to any metric space, but this is trivial.}} that
 \[
  \tf^-(p) = \int_0^1 \vol^-_r(p) \, dr
 \]
 is continuous at $q$.
\end{proof}

%%%%%%%%%%%%%%%%%%%%%%%%%%%%%% GLOB HYP %%%%%%%%%%%%%%%%%%%%%%%%%%%%%%%%%%%%%%

\section{Global hyperbolicity and Cauchy time functions} \label{ghsec}

The highest step on the causal ladder, namely global hyperbolicity, is fundamental for a number of deep and important results in general relativity, such as the study of the Cauchy problem of the Einstein equations, the singularity theorems, and Lorentzian splitting theorems. This is due to the fact global hyperbolicity is equivalent to the existence of a Cauchy time function, whose level sets in turn are Cauchy surfaces, i.e., domains suitable for specifying initial data for hyperbolic PDEs, and for imposing focusing conditions for geodesics. This characterization of global hyperbolicity was first obtained by Geroch~\cite{Ger} in 1970 and makes use of volume time functions. In the same vein, in this section we characterize global hyperbolicity for Lorentzian pre-length spaces in four different ways by also utilizing our constructions from Section~\ref{volfcts}.

\subsection{Definitions and main result}

The causality conditions we use for Lorentzian pre-length space are defined analogously to the smooth case as follows (cf.\ \cite{Min3}).

\begin{defn}[{\cite[Definition 2.35]{KuSa}}]\label{ghdef}
 A Lorentzian pre-length space $(X,d,\tll, \allowbreak \leq,\tau)$ is called \emph{globally hyperbolic} if it is non-totally imprisoning and the causal diamonds $J(p,q)$ are compact for all $p,q \in X$.
\end{defn}

\begin{defn}
 A time function $t \colon X \to \real$ on a Lorentzian pre-length space $(X,d,\tll,\leq,\tau)$ is called a \emph{Cauchy time function} if for every doubly-inextendible causal curve $\gamma$ we have $\operatorname{Im}(t \circ \gamma) = \real$.
\end{defn}

For smooth and continuous Lorentzian metrics, global hyperbolicity is also characterized by the existence of a Cauchy surface, which is then a topological (even smooth) hypersurface. We extend the definition verbatim, but adopt the name \emph{Cauchy set} to emphasize that we are not in the manifold setting.

\begin{defn} \label{defcahy}
 Let $(X,d,\tll,\leq,\tau)$ be a Lorentzian pre-length space. A \emph{Cauchy set} is a subset $\cahy \subseteq X$ such that every doubly-inextendible causal curve intersects $\cahy$ exactly once.
\end{defn}

Geroch~\cite[Section 4]{Ger} makes use of Leray's notion of global hyperbolicity which is formulated in terms of the topology on the collection of certain curves. We will show that Definition~\ref{ghdef} is equivalent to this notion also for Lorentzian pre-length spaces. To this end, for any two points $p,q \in X$, we consider the set $\cpq(p,q)$ the equivalence class of future-directed causal curves from $p$ to $q$ with continuous, strictly monotonically increasing reparametrizations, equipped with the Hausdorff distance between the images of the curves as subsets in $X$, i.e.,
\[
 d_H(\gamma_1,\gamma_2) = \max \{ \sup_{x \in \gamma_1}d(x,\gamma_2), \sup_{y \in \gamma_2}d(\gamma_1,y) \}
\]
(we write $\gamma_i$ also for the image $\operatorname{Im}(\gamma_i)$, since the parametrization does not matter).

In this section we establish the third, and last, main result of this paper (which via Proposition \ref{LLSareapprox} immediately implies the Lorentzian length version stated in Theorem~\ref{thm3} of the Introduction).

\begin{thm} \label{ghthm}
 Let $(X,d,\tll,\leq,\tau)$ be an approximating Lorentzian pre-length space with limit curves. Suppose, in addition, that $(X,d)$ is second countable and proper. Then the following are equivalent:
 \begin{enumerate}
  \item $X$ is globally hyperbolic,
  \item $X$ is non-totally imprisoning and $\cpq(p,q)$ is compact, for any $p,q \in X$,
  \item $X$ admits a Cauchy time function,
  \item $X$ admits a Cauchy set.
 \end{enumerate}
\end{thm}

\begin{rem}[Smooth spacetimes]
 Manifolds are second-countable by definition. Moreover, any differentiable manifold admits a complete Riemannian metric \cite{NoOz}, and hence, by the Hopf--Rinow Theorem, a proper distance. For spacetimes with continuous metrics, however, Theorem~\ref{ghthm} cannot be applied unrestrictedly, since only causally plain $C^0$-spacetimes satisfy the axioms of Lorentzian pre-length spaces \cite[Example 5.2]{KuSa} (but a characterization of global hyperbolicity on \emph{all} $C^0$-spacetimes has been obtained in \cite{Sae}). On the other hand, Theorem~\ref{ghthm} is of course valid well beyond the manifold setting. For instance, by the Hopf--Rinow--Cohn-Vossen Theorem it is sufficient that $(X,d)$ is a complete, locally compact length metric space for it to be proper.
\end{rem}

\begin{rem}[Topology change]
 On a spacetime $(M,g)$, if $\cahy$ is a smooth Cauchy surface, then any other Cauchy surface is diffeomorphic to $\cahy$ \cite{BeSa1} (similarly, if one works with continuous Cauchy surfaces, then they are homeomorphic). The diffeomorphism can be constructed by following the flow of the time-orientation vector field. Further, $M$ is foliated by Cauchy surfaces.

 For Lorentzian pre-length spaces $X$, it is still true that in the setting of Theorem~\ref{ghthm} the level sets of Cauchy time functions yield a decomposition of $X$ as a disjoint union of Cauchy sets. Different Cauchy sets, however, need not be homeomorphic, nor even homotopy equivalent as the next examples show.
\end{rem}

\begin{ex}[Degenerate subsets of Minkowski space]
 Let $(M,\eta)$ be $(n+1)$-dimensional Minkowski spacetime, with coordinates $(t,x) \in \real \times \real^n$ and $n \geq 2$. Define
 \[ X := \{(t,x) \mid t\leq0, x=0 \text{ or } t>0, \vert x \vert =t/2\} \]
 (here $\vert x \vert$ denotes the Euclidean norm of $x\in\real^n$), equipped with the causal and chronological relations induced by $\eta$ and the Euclidean distance. Then $X$ is a globally hyperbolic Lorentzian pre-length space and satisfies the assumptions of Theorem~\ref{ghthm}. The function $f(t,x) = t$ is a Cauchy time function, with some of its level sets being points (if $t\leq0$) and some being $(n-1)$-spheres (if $t>0$).
\end{ex}

\begin{ex}[Degenerate generalized cones]
 If $(X_1,d_1)$, $(X_2,d_2)$ are separable, proper geodesic length spaces one can construct generalized cones $Y_1$ and $Y_2$ in the sense of \cite{AGKS} over them and glue them together at the tip. The resulting space can be equipped with the structure of a Lorentzian length space in the usual way, which is then globally hyperbolic by \cite[Prop.\ 4.10]{AGKS}, and has Cauchy sets homeomorphic to $X_1$, $X_2$ and $\{ \text{tip} \}$.
\end{ex}

\subsection{Properties of Cauchy sets}\label{ssec:cahy}

In what follows we prove several results involving Cauchy sets that are crucial in the proof of Theorem~\ref{ghthm} for the implication (iv) $\Longrightarrow$ (ii), but are also of independent interest. We make use of several results about inextendible causal curves, in particular, about the existence of a maximal extension of causal curves obtained in Section~\ref{ssec:def}.

\begin{prop}[Basic properties of Cauchy sets]\label{cahyprop}
 Let $(X,d,\ll,\leq,\tau)$ be an approximating Lorentzian pre-length space with limit curves such that $(X,d)$ is proper. If $\cahy \subseteq X$ is a Cauchy set, then the following properties hold:
 \begin{enumerate}
  \item $\cahy$ is \emph{acausal}, i.e., distinct points on $\cahy$ are not causally related,
  \item $X = J^-(\cahy) \cup J^+(\cahy)$,
  \item $X = I^-(\cahy) \sqcup \cahy \sqcup I^+(\cahy)$, where $\sqcup$ denotes the disjoint union,
  \item $J^\pm(\cahy) = \overline{I^\pm(\cahy)} = I^\pm(\cahy) \sqcup \cahy$, and $\cahy$ and $J^\pm(\cahy)$ are closed.
 \end{enumerate}
\end{prop}

\begin{proof}
 (i) Suppose $p < q$ and $p,q \in \cahy$. Then by causal path-connectedness, there exists a causal curve $\gamma \colon [a,b] \to X$ from $p$ to $q$. By Proposition~\ref{exisinex} there exists a maximal doubly-inextendible extension $\lambda$ of $\gamma$, which therefore intersects $\cahy$ more than once, a contradiction to $\cahy$ being a Cauchy set.
 
 (ii) Let $p \in X$. By Corollary~\ref{decomp1}, there exists a doubly-inextendible causal curve $\gamma$ through $p$. Since $\cahy$ is a Cauchy set, it intersects $\gamma$ exactly once, say at $\gamma(0)$. Hence $p \leq \gamma(0)$ or $\gamma(0) \leq p$.
 
 (iii) Let $p \in X$. By Lemma~\ref{pathapprox}, we can assume that the curve $\gamma$ in (ii) is timelike in a neighborhood of $p$. Then, if $p \leq \gamma(0)$, we have either $p = \gamma(0) \in \cahy$ or $p \tll \gamma(-\epsilon) < \gamma(0)$ for some $\epsilon > 0$, which then by the push-up Lemma~\ref{push-up} implies $p \tll \gamma(0)$. Applying the same argument in the case $\gamma(0) \leq p$, we conclude that $p \in I^-(\cahy) \cup \cahy \cup I^+(\cahy)$.
 
 It remains to be shown that the union is disjoint. Since $\cahy$ is acausal by (i) and the push-up property it is clear that $\cahy$ and $I^\pm(\cahy)$ are disjoint. Moreover, $I^+(\cahy)$ and $I^-(\cahy)$ are disjoint since otherwise, there exists a causal curve that intersects $\cahy$ at two different points or is closed timelike (if the points coincide), and hence contradicts $\cahy$ being Cauchy.
 
 (iv) The sets $I^\pm(\cahy) = \bigcup_{p \in \cahy} I^\pm(p)$ are unions of open sets and hence open. By (iii) is $\cahy = X \setminus (I^-(\cahy) \sqcup I^+(\cahy))$ the complement of open sets and thus closed. By the same argument is $X \setminus I^\mp(\cahy) = I^\pm(\cahy) \sqcup \cahy$ closed, and hence
 \[
 \overline{I^\pm(\cahy)} \subseteq \overline{I^\pm(\cahy) \sqcup \cahy} = I^\pm(\cahy) \sqcup \cahy \subseteq J^\pm(\cahy) \subseteq \overline{I^\pm(\cahy)},
 \]
 where the last inclusion is due to the assumption that $X$ is approximating. Thus $I^\pm(\cahy) \sqcup \cahy = J^\pm(\cahy) = \overline{I^\pm(\cahy)}$, and hence closed.
\end{proof}

Closedness of Cauchy sets, in particular, will be key in establishing (iv) $\Longrightarrow$ (ii) of Theorem~\ref{ghthm}. In the remaining lemmas of this subsection we prove this implication in steps.

\begin{lem} \label{nti}
 Let $(X,d,\tll,\leq,\tau)$ be an approximating Lorentzian pre-length space with limit curves. Suppose, in addition, that $(X,d)$ is proper. If $X$ contains a Cauchy set $\cahy$, then $X$ is non-totally imprisoning.
\end{lem}

\begin{proof}
 Suppose $X$ is not non-totally imprisoning. Then by Theorem~\ref{cor315}, there is a compact set $K \subseteq X$ and a doubly-inextendible future-directed causal curve $\gamma \colon \real \to K$. By Definition \ref{defcahy} we know that, without loss of generality, $\gamma \cap \cahy = \{\gamma(0)\}$.
 
 First we show that $\tilde{\gamma} := \gamma \vert_{[1,\infty)}$ must come arbitrarily close to $\cahy$, that is, there exist parameter values $s_n$ such that $d(\tilde{\gamma}(s_n),\cahy) < 1/n$, for all $n \in \nat$. Suppose this is not the case: then there exists a $\delta$ such that $d(\tilde{\gamma}(s),\cahy) > \delta$ for all $s \in [1,\infty)$, and hence $\tilde{\gamma}$ is contained in the compact set $K \setminus \cahy_\delta$, where $\cahy_\delta := \{ x \in X \mid d(x,\cahy) < \delta \}$. But then, since $\tilde\gamma$ is future-inextendible, by the proof of Theorem~\ref{cor315}, there also exists a doubly-inextendible causal curve contained in $K \setminus \cahy_\delta$, in contradiction to $\cahy$ being a Cauchy set. Thus we conclude that there must exist a sequence $(s_n)_n$ such that $d(\tilde{\gamma}(s_n),\cahy) \to 0$. Moreover, since $\cahy$ is closed by Proposition~\ref{cahyprop} (iv), $K \cap \cahy$ is compact, and hence $\left(\tilde{\gamma}(s_n)\right)_n$ converges, up to a subsequence, to some point $p \in \cahy$. This implies that $s_n \to \infty$, because if $s_\infty = \lim s_n \in [1,\infty)$, then $\gamma(s_\infty) \in \cahy$, which is not allowed since already $\gamma(0) \in \cahy$ and $\cahy$ is a Cauchy set. 

 Having established the existence of a sequence $s_n \to \infty$ in $\real$ such that $\gamma(s_n) \to p \in \cahy$, we may consider the sequence of past-directed causal curves $\lambda_n : [0,s_n] \to X$ given by
 \[
  \lambda_n(s) = \gamma (s_n - s).
 \]
 By Lemma \ref{lem312}, $\gamma$ has infinite arclength, and thus $L^d(\lambda_n) \to \infty$. Since $\lambda_n(0) \to p$, we can apply Theorem \ref{thm314} to find a past-directed limit curve $\lambda : [0,\infty) \to X$ with $\lambda(0) = p$, hence $\lambda([0,\infty)) \subseteq J^-(p) \subseteq J^-(\cahy)$. On the other hand, we have that $\gamma([0,\infty)) \subseteq J^+(\gamma(0)) \subseteq J^+(\cahy)$, therefore $\lambda_n([0,s_n]) \subseteq J^+(\cahy)$. By Proposition~\ref{cahyprop}(iv), $J^+(\cahy)$ is closed, so as a limit curve, $\lambda$ is contained in it. But then $\lambda([0,\infty)) \subseteq J^-(\cahy) \cap J^+(\cahy)$. This is a contradiction since, by Proposition~\ref{cahyprop} (iv), $J^-(\cahy) \cap J^+(\cahy) = \cahy$, and by Proposition~\ref{cahyprop} (i), Cauchy surfaces are acausal.
\end{proof}

Finally, the last two lemmas prove compactness of $\cpq(p,q)$. They are adapted from Geroch's original proof for smooth spacetimes \cite{Ger}. Similarly to $\cpq(p,q)$, we denote by $\cpq(p,\cahy)$ the space of future-directed causal curves from a point $p \in X$ to a Cauchy set $\cahy$, equipped with the Hausdorff distance.

\begin{lem} \label{CpS}
 Let $(X,d,\tll,\leq,\tau)$ be an approximating Lorentzian pre-length space with limit curves. Suppose, in addition, that $(X,d)$ is proper. Let $\cahy$ be a Cauchy surface in $X$. Then for all $p \in X$, the space $\cpq(p,\cahy)$ is compact.
\end{lem}

\begin{proof}
 Suppose $\cpq(p,\cahy) \neq \emptyset$, otherwise the statement is trivial. In particular, this means that $p \in J^-(\cahy)$. The set $\cpq(p,\cahy)$ together with the Hausdorff distance $d_H$ is a metric space, and hence it remains to prove sequential compactness. Let $(\gamma_n)_n$ be a sequence in $\cpq(p,\cahy)$. We distinguish between two cases:
 \begin{enumerate}
  \item[1.] There exists a constant $C>0$ such that $L^d(\gamma_n) < C$ for all $n \in \nat$. Assume that all $\gamma_n$ are defined on $[0,1]$. Then by Theorem \ref{thm37}, a subsequence of $(\gamma_n)_n$ converges uniformly (and thus also with respect to $d_H$) to a future-directed causal limit curve $\gamma \colon [0,1] \to X$. Since $\gamma_n(0) = p$, also $\gamma(0) = p$. Moreover, $\gamma_n(1) \in \cahy$ for all $n \in \nat$, and $\cahy$ is closed by Proposition~\ref{cahyprop} (iv), thus also $\gamma(1) \in \cahy$. Hence $\gamma \in \cpq(p,\cahy)$.
  \item[2.] On the other hand, assume that $L^d(\gamma_n) \to \infty$ and all $\gamma_n$ are parametrized by $d$-arclength. By Theorem~\ref{thm314} a subsequence of $(\gamma_n)_n$ converges locally uniformly to a future-inextendible causal limit curve $\gamma \colon [0,\infty) \to X$. Since $\gamma_n \subseteq J^-(\cahy)$, and by Proposition~\ref{cahyprop}(iv), $J^-(\cahy)$ is closed, we have $\gamma \subseteq J^-(\cahy)$. By Proposition~\ref{exisinex}, there exists a doubly-inextendible extension $\tilde\gamma \colon \real \to X$ of $\gamma$, which by transitivity of $\leq$ is also contained in $J^-(\cahy)$. Then $\tilde\gamma$ must intersect $\cahy$, say at $\tilde\gamma (s_0)$. This implies that for all $s > s_0$, $\gamma(s) \in J^-(\cahy) \cap J^+(\cahy) = \cahy$. Moreover, $\gamma\vert_{(s_0,\infty)}$ cannot be constant, because then by Lemma~\ref{lem312}, $\gamma$ would be future-extendible. Since by Proposition~\ref{cahyprop} (i) $\cahy$ is acausal, we arrive at a contradiction. \qedhere
 \end{enumerate}
\end{proof}

\begin{lem} \label{chcpqcompact}
 Let $(X,d,\tll,\leq,\tau)$ be an approximating Lorentzian pre-length space with limit curves. Suppose, in addition, that $(X,d)$ is proper. If $X$ contains a Cauchy set $\cahy$, then $\cpq(p,q)$ is compact for all $p,q \in X$.
\end{lem}

\begin{proof}
 By Lemma~\ref{nti}, if $X$ contains a Cauchy set  $\cahy$, then $X$ is non-totally imprisoning and hence causal. Thus we may assume without loss of generality that $p < q$ so that $\cpq(p,q)$ is nontrivial. We prove sequential compactness, with $(\gamma_n)_n$ always denoting a sequence in $\cpq(p,q)$. In view of Proposition~\ref{cahyprop} (i) and (iii), we can distinguish three cases:
 \begin{itemize}
  \item[1.] $p \in \cahy$ and $q \in I^+(\cahy)$ (or analogously, $q \in \cahy$ and $p \in I^-(\cahy)$): Then $(\gamma_n)_n$ can be seen as a sequence in $\cpq(\cahy,q)$. By Lemma \ref{CpS}, there exists a limit curve $\gamma$ in $\cpq(\cahy,q)$. But since $\gamma_n$ starts at $p$ for all $n$, also $\gamma$ must start at $p$, hence it is an element of $\cpq(p,q)$.
  \item[2.] $p,q \in I^+(\cahy)$ (or analogously, $p,q \in I^-(\cahy)$): There exists a future-directed timelike curve $\lambda$ from $\cahy$ to $p$. Construct a new sequence $(\tilde\gamma_n)_n$ by concatenating $\lambda$ with $\gamma_n$. Then $(\tilde\gamma_n)_n$ is a sequence in $\cpq(\cahy,q)$ and thus has a causal limit curve $\tilde\gamma$ by Lemma~\ref{CpS}. Because of how the sequence was constructed, $\tilde\gamma$ must be the concatenation of $\lambda$ with a causal curve $\gamma \in \cpq(p,q)$ which is the Hausdorff limit of $(\gamma_n)_n$.
  \item[3.] $p \in I^-(\cahy)$ and $q \in I^+(\cahy)$: By Proposition~\ref{exisinex} we can extend each $\gamma_n$, and the maximal extension must intersect $\cahy$ exactly once, say at $\gamma_n(0)$. Because of Proposition~\ref{cahyprop} (iv), $\gamma_n(0)$ cannot lie to the past of $p$ nor to the future of $q$, hence it must in fact lie on $\gamma_n$. Consider the sequence $(\bar\gamma_n)_n$ where $\bar\gamma_n$ is the restriction of $\gamma_n$ from $p$ to $\gamma_n(0)$. By Lemma~\ref{CpS} a subsequence $(\bar\gamma_{n_k})_k$ converges to a limit curve $\bar\gamma$ in $\cpq(p,\cahy)$. Similarly, consider the sequence $(\tilde\gamma_{n_k})_k$ of the restrictions of the original curves $\gamma_{n_k}$ from $\gamma_{n_k}(0)$ to $q$. By Lemma~\ref{CpS}, we may assume it converges to a limit curve $\tilde\gamma$ in $\cpq(\cahy,q)$. Since by construction the endpoints of $\bar\gamma$ and $\tilde\gamma$ agree on $\cahy$, we can join them to obtain a limit causal curve of (a subsequence of) the original sequence $(\gamma_n)_n$ from $p$ to $q$. \qedhere
 \end{itemize}
\end{proof}

\subsection{Proof of Theorem~\ref{ghthm}}

We prove Theorem~\ref{ghthm} in several steps. The implications (ii) $\Longrightarrow$ (i) and (iii) $\Longrightarrow$ (iv) are straightforward. The most involved step (iv) $\Longrightarrow$ (ii) follows from our results in Section~\ref{ssec:cahy} about properties of Cauchy sets. Finally, for the implication (i) $\Longrightarrow$ (iii) we show that the averaged volume functions of Section~\ref{volfcts} have additional properties on globally hyperbolic Lorentzian pre-length spaces.

\begin{proof}[Proof of Theorem~\ref{ghthm}]
 (ii) $\Longrightarrow$ (i) 
 By (ii), $X$ is already non-totally imprisoning, and thus it remains to be shown that the causal diamonds $J(p,q)$ are compact for all $p\leq q$. If $p = q$, then $J(p,q) = \{p\}$ because non-total imprisonment implies causality. Suppose that $p<q$. Let $(x_n)_n$ be any sequence in $J(p,q)$. By causal path-connectedness of $X$, every $x_n$ lies on a causal curve $\gamma_n \colon [0,1] \to X$ from $p$ to $q$. By (ii), the space of curves $\cpq(p,q)$ is compact, and hence a subsequence $(\gamma_{n_k})_k$ of $(\gamma_n)_n$ converges to a causal curve $\gamma \in \cpq(p,q)$ in the Hausdorff sense. In particular, for the corresponding points, $d(x_{n_k},\gamma) \to 0$ as $k \to \infty$. Since $\gamma$ itself is compact, a subsequence of $(x_{n_k})_k$ must converge to a point on $\gamma$. Hence $J(p,q)$ is compact.
 
 (iii) $\Longrightarrow$ (iv) Suppose $t$ is a Cauchy time function. Then the level sets of $t$ are Cauchy sets: For $s \in \real$ consider the preimage $\cahy = t^{-1}(\{ s \})$. Since $t$ is a time function, any future-directed causal curve $\gamma$ intersects $\cahy$ at most once. If $\gamma$ is furthermore doubly-inextendible, and since $t$ is Cauchy, $\ima(t \circ \gamma) = \real$ and thus $\cahy \cap \gamma \neq \emptyset$. Hence doubly-inextendible causal curves intersect $\cahy$ at exactly one point.
 
 (iv) $\Longrightarrow$ (ii) Suppose $X$ admits a Cauchy set. By Lemma~\ref{nti}, $X$ is non-totally imprisoning. By Lemma~\ref{chcpqcompact}, the set of future-directed causal curves $\cpq(p,q)$ between $p$ and $q$ is compact for any $p,q \in X$.
 
 (i) $\Longrightarrow$ (iii)
 Let $t^\pm$ be given by Definition~\ref{volfct} and define
 \[ 
  t := \ln\left(-\frac{t^-}{t^+}\right).
 \]
 By (i), $X$ is globally is hyperbolic, and thus it follows immediately from Lemma~\ref{ghvolfcts} (see below) that $t$ is a Cauchy time function. Note that here we are using the assumption of second countability in order to equip $(X,d)$ with a Borel probability measure of full support (see Proposition \ref{sep}).
\end{proof}

We finish the remaining parts of the proof of (i) $\Longrightarrow$ (iii). By definition, $I^\pm(p) \subseteq J^\pm(p)$ for all points $p$ in a Lorentzian pre-length space $X$. If $X$ is approximating then furthermore $J^\pm(p) \subseteq \overline{I^\pm(p)}$, and thus $\overline{I^\pm(p)} = \overline{J^\pm(p)}$. If $X$ is globally hyperbolic, we can say even more, which will allow us to apply our results of Section~\ref{volfcts}.

\begin{prop} \label{gh-cc}
 Let $(X,d,\tll,\leq,\tau)$ be a causally path-connected, approximating Lorentzian pre-length space.
 If $X$ is globally hyperbolic, then $X$ is \emph{causally simple\/}, meaning that $J^\pm(p)$ is closed and thus $J^\pm(p) = \overline{I^\pm(p)}$ for all $p \in X$. Moreover, if $X$ is causally simple, then $X$ is causally continuous.
\end{prop}

\begin{proof}
 These statements were shown by Ak\'e et al.\ for Lorentzian length spaces \cite[Propositions 3.13 \& 3.14]{ACS}. The same proof goes through for our assumptions, because the assumption of localizability is only needed to invoke \cite[Sequence Lemma 2.25]{ACS} which in our case is replaced by Lemma~\ref{lemapprox}.
\end{proof}

We also need the following result.

\begin{lem} \label{ghvolfcts}
 Let $(X,d,\tll,\leq,\tau)$ be an approximating Lorentzian pre-length space with limit curves. Suppose, in addition, that $(X,d)$ is a proper metric space equipped with a Borel probability measure of full support and corresponding averaged volume functions $t^+ \colon X \to [-\infty,0]$ and $t^- \colon X \to [0,\infty]$ (see Definition~\ref{volfct}). If $X$ is globally hyperbolic, then $t^\pm$ are time functions. Moreover, for every doubly-inextendible future-directed causal curve $\gamma \colon (a,b) \to X$ it holds that
  \[
   \lim_{s \to b} t^+(\gamma(s)) = \lim_{s \to a} t^-(\gamma(s)) = 0.
  \]
\end{lem}

\begin{proof}
 We follow \cite[Satz II.20]{Die2}. Suppose $X$ is a globally hyperbolic, approximating Lorentzian pre-length space with limit curves. Then by Proposition~\ref{gh-cc} $X$ is causally continuous, and thus by Theorem~\ref{ccthm} $t^\pm$ are time functions.
 
 To show the second part of the statement, assume for contradiction that $\gamma\colon (a,b] \to X$ is a future-directed past-inextendible causal curve with
 \begin{equation} \label{t-g0}
  \lim_{s \to a} t^-(\gamma(s)) > 0.
 \end{equation}
 On the other hand, by standard measure theory \cite[Thm.\ 1.2.5]{KrPa} and Lebesgue's dominated convergence theorem,
 \begin{align*}
  \lim_{s \to a} t^-(\gamma(s)) &= \lim_{s \to a} \int_0^1 \mu\left(I^-_r(\gamma(s))\right) dr \\
  &= \int_0^1 \lim_{s \to a} \mu\left(I^-_r(\gamma(s))\right) dr \\
  &= \int_0^1 \mu\left(\bigcap_{s>a} I^-_r(\gamma(s)) \right) dr.
 \end{align*}
 Assumption~\eqref{t-g0} implies that there exists an $r \in (0,1)$ and a point \[p \in \bigcap_{s>a} I^-_r(\gamma(s)).\] This means that for every $s>a$, there exists a $q \in I^-(\gamma(s))$ such that $d(q,p) < r$. In particular, we can find a sequence $q_n \in I^-\left(\gamma(a + 1/n)\right)$ such that $d(q_n,p) < r$ for all $n \in \nat$. Because $d$ is proper, the closed ball of radius $r$ around $p$ is compact, and hence $q_n$ converges to a limit point $q$ (up to a subsequence). Now for a given $s > a$, we choose $n_0$ such that $s > a + 1/n_0$. Then, by the push-up Lemma~\ref{push-up}, $q_n \in I^-(\gamma(s))$ for all $n \geq n_0$. Because of this, $q \in \overline{I^-(\gamma(s))} = J^-(\gamma(s))$, where the last equality follows by global hyperbolicity and Proposition \ref{gh-cc}. Since $s>a$ was arbitrary, we have that $\gamma \subseteq J^+(q)$. In particular, $q \leq \gamma(b)$, and by global hyperbolicity, $J(q,\gamma(b))$ is compact. But then the curve $\gamma$ is imprisoned in $J(q,\gamma(b))$, in contradiction to global hyperbolicity.
\end{proof}

This concludes our characterization of global hyperbolicity with the existence of Cauchy time functions (and Cauchy sets) in the setting of Lorentzian pre-length spaces.

\section*{Acknowledgements}

LGH would like to thank Didier Solis for communication related to their work~\cite{ACS}. The authors are also grateful to Ettore Minguzzi for feedback on the introduction, and to the anonymous referees for constructive comments. AB's research is supported by the Dutch Research Council (NWO), Project number VI.Veni.192.208.
\medskip

\noindent \textbf{Declarations.} The authors state that there is no conflict of interest. No data was collected or analysed as part of this project. This version of the article has been accepted for publication after peer review, but is not the Version of Record and does not reflect post-acceptance improvements, or any corrections. The Version of Record is available online at: \href{http://dx.doi.org/10.1007/s00023-024-01461-y}{http://dx.doi.org/10.1007/s00023-024-01461-y}

\bibliographystyle{abbrv}
\bibliography{bbib}

\def\cprime{$'$}
\begin{thebibliography}{10}

\bibitem{ACS}
L.~{Aké Hau}, A.~J. {Cabrera Pacheco}, and D.~A. Solis.
\newblock On the causal hierarchy of {L}orentzian length spaces.
\newblock {\em Classical Quantum Gravity}, 37(21):215013, 2020.

\bibitem{AGKS}
S.~B. Alexander, M.~Graf, M.~Kunzinger, and C.~S{\"{a}}mann.
\newblock Generalized cones as {L}orentzian length spaces: Causality,
  curvature, and singularity theorems.
\newblock To appear in {\it Comm. Anal. Geom.} arXiv:1909.09575.

\bibitem{AlBu}
B.~Allen and A.~Burtscher.
\newblock Properties of the null distance and spacetime convergence.
\newblock {\em Int. Math. Res. Not. IMRN}, (10):7729–7808, 2022.

\bibitem{AGH}
L.~Andersson, G.~J. Galloway, and R.~Howard.
\newblock The cosmological time function.
\newblock {\em Classical Quantum Gravity}, 15(2):309–322, 1998.

\bibitem{BGP}
C.~{B{\"{a}}r}, N.~{Ginoux}, and F.~{Pf{\"{a}}ffle}.
\newblock {\em Wave equations on Lorentzian manifolds and quantization.}
\newblock Z{\"{u}}rich: European Mathematical Society Publishing House, 2007.

\bibitem{Ba}
R.~Bartnik.
\newblock Remarks on cosmological spacetimes and constant mean curvature
  surfaces.
\newblock {\em Comm. Math. Phys.}, 117(4):615–624, 1988.

\bibitem{BeSa1}
A.~N. Bernal and M.~S{\'{a}}nchez.
\newblock On smooth {C}auchy hypersurfaces and {G}eroch's splitting theorem.
\newblock {\em Comm. Math. Phys.}, 243(3):461–470, 2003.

\bibitem{BeSa}
A.~N. Bernal and M.~S{\'{a}}nchez.
\newblock Smoothness of time functions and the metric splitting of globally
  hyperbolic spacetimes.
\newblock {\em Comm. Math. Phys.}, 257(1):43–50, 2005.

\bibitem{BeSa3}
A.~N. Bernal and M.~S{\'{a}}nchez.
\newblock Further results on the smoothability of {C}auchy hypersurfaces and
  {C}auchy time functions.
\newblock {\em Lett. Math. Phys.}, 77(2):183–197, 2006.

\bibitem{BeSu}
P.~Bernard and S.~Suhr.
\newblock Lyapounov functions of closed cone fields: from {C}onley theory to
  time functions.
\newblock {\em Comm. Math. Phys.}, 359(2):467–498, 2018.

\bibitem{BeSu2}
P.~Bernard and S.~Suhr.
\newblock Cauchy and uniform temporal functions of globally hyperbolic cone
  fields.
\newblock {\em Proc. Amer. Math. Soc.}, 148(11):4951–4966, 2020.

\bibitem{BLMS}
L.~Bombelli, J.~Lee, D.~Meyer, and R.~Sorkin.
\newblock {Space-Time as a Causal Set}.
\newblock {\em Phys. Rev. Lett.}, 59:521–524, 1987.

\bibitem{BoSe}
H.-J. Borchers and R.~N. Sen.
\newblock Theory of ordered spaces.
\newblock {\em Comm. Math. Phys.}, 132(3):593–611, 1990.

\bibitem{BoSe2}
H.-J. Borchers and R.~N. Sen.
\newblock {\em Mathematical implications of {E}instein-{W}eyl causality},
  volume 709 of {\em Lecture Notes in Physics}.
\newblock Springer-Verlag, Berlin, 2006.

\bibitem{BDGSS}
A.~Borde, H.~F. Dowker, R.~S. Garcia, R.~D. Sorkin, and S.~Surya.
\newblock Causal continuity in degenerate spacetimes.
\newblock {\em Classical Quantum Gravity}, 16(11):3457–3481, 1999.

\bibitem{BBI}
D.~Burago, Y.~Burago, and S.~Ivanov.
\newblock {\em A course in metric geometry}, volume~33 of {\em Graduate Studies
  in Mathematics}.
\newblock American Mathematical Society, Providence, RI, 2001.

\bibitem{Bus}
H.~Busemann.
\newblock Timelike spaces.
\newblock {\em Dissertationes Math. (Rozprawy Mat.)}, 53:52, 1967.

\bibitem{CaMo}
F.~Cavalletti and A.~Mondino.
\newblock Optimal transport in {L}orentzian synthetic spaces, synthetic
  timelike {R}icci curvature lower bounds and applications.
\newblock To appear in \textit{Camb.\ J.\ Math.} arXiv:2004.08934.

\bibitem{ChGr}
P.~T. Chru{\'{s}}ciel and J.~D.~E. Grant.
\newblock On {L}orentzian causality with continuous metrics.
\newblock {\em Classical Quantum Gravity}, 29(14):145001, 32, 2012.

\bibitem{ChGrMi}
P.~T. Chru{\'{s}}ciel, J.~D.~E. Grant, and E.~Minguzzi.
\newblock On differentiability of volume time functions.
\newblock {\em Ann. Henri Poincar{\'{e}}}, 17(10):2801–2824, 2016.

\bibitem{Die2}
J.~Dieckmann.
\newblock {\em Volumenfunktionen in der allgemeinen Relativit{\"{a}}tstheorie}.
\newblock PhD thesis, Berlin, 1987.

\bibitem{Die3}
J.~Dieckmann.
\newblock Cauchy surfaces in a globally hyperbolic space-time.
\newblock {\em J. Math. Phys.}, 29(3):578–579, 1988.

\bibitem{Die}
J.~Dieckmann.
\newblock Volume functions in general relativity.
\newblock {\em Gen. Relativity Gravitation}, 20(9):859–867, 1988.

\bibitem{EPS}
J.~Ehlers, F.~A.~E. Pirani, and A.~Schild.
\newblock The geometry of free fall and light propagation.
\newblock In {\em General relativity (papers in honour of {J}. {L}. {S}ynge)},
  page 63–84. Clarendon Press, Oxford, 1972.

\bibitem{Es}
J.-H. Eschenburg.
\newblock The splitting theorem for space-times with strong energy condition.
\newblock {\em J. Differential Geom.}, 27(3):477–491, 1988.

\bibitem{Fa}
A.~Fathi.
\newblock Time functions revisited.
\newblock {\em Int. J. Geom. Methods Mod. Phys.}, 12(8):1560027, 12, 2015.

\bibitem{FaSi}
A.~Fathi and A.~Siconolfi.
\newblock On smooth time functions.
\newblock {\em Math. Proc. Cambridge Philos. Soc.}, 152(2):303–339, 2012.

\bibitem{Fin}
F.~{Finster}.
\newblock Causal fermion systems: A primer for lorentzian geometers.
\newblock In {\em Journal of Physics Conference Series}, volume 968 of {\em
  Journal of Physics Conference Series}, page 012004, 2018.

\bibitem{Ga}
G.~J. Galloway.
\newblock The {L}orentzian splitting theorem without the completeness
  assumption.
\newblock {\em J. Differential Geom.}, 29(2):373–387, 1989.

\bibitem{GH}
L.~Garc{\'{i}}a-Heveling.
\newblock Causality theory of spacetimes with continuous {L}orentzian metrics
  revisited.
\newblock {\em Classical Quantum Gravity}, 38(14):145028, 13, 2021.

\bibitem{Ger}
R.~Geroch.
\newblock Domain of dependence.
\newblock {\em J. Math. Phys.}, 11:437–449, 1970.

\bibitem{GerTopo}
R.~P. Geroch.
\newblock Topology in general relativity.
\newblock {\em J. Math. Phys.}, 8:782–786, 1967.

\bibitem{GKS}
J.~D.~E. Grant, M.~Kunzinger, and C.~S{\"{a}}mann.
\newblock Inextendibility of spacetimes and {L}orentzian length spaces.
\newblock {\em Ann. Global Anal. Geom.}, 55(1):133–147, 2019.

\bibitem{GKSS}
J.~D.~E. Grant, M.~Kunzinger, C.~S{\"{a}}mann, and R.~Steinbauer.
\newblock The future is not always open.
\newblock {\em Lett. Math. Phys.}, 110(1):83–103, 2020.

\bibitem{Haw}
S.~Hawking.
\newblock The existence of cosmic time functions.
\newblock {\em Proc. Roy. Soc. Lond. A}, A308:433–435, 1968.

\bibitem{HaEl}
S.~W. Hawking and G.~F.~R. Ellis.
\newblock {\em {The large scale structure of space-time}}.
\newblock Cambridge Monographs on Mathematical Physics. Cambridge University
  Press, 1973.

\bibitem{HaSa}
S.~W. Hawking and R.~K. Sachs.
\newblock Causally continuous spacetimes.
\newblock {\em Comm. Math. Phys.}, 35:287–296, 1974.

\bibitem{Hor}
G.~T. Horowitz.
\newblock Topology change in classical and quantum gravity.
\newblock {\em Classical Quantum Gravity}, 8:587–602, 1991.

\bibitem{Jost}
J.~Jost.
\newblock {\em Postmodern analysis}.
\newblock Universitext. Springer-Verlag, Berlin, third edition, 2005.

\bibitem{KrPa}
S.~G. Krantz and H.~R. Parks.
\newblock {\em Geometric integration theory}.
\newblock Cornerstones. Birkhäuser Boston, 2008.

\bibitem{KrPe}
E.~H. Kronheimer and R.~Penrose.
\newblock On the structure of causal spaces.
\newblock {\em Proc. Cambridge Philos. Soc.}, 63:481–501, 1967.

\bibitem{KuSa}
M.~Kunzinger and C.~S{\"{a}}mann.
\newblock Lorentzian length spaces.
\newblock {\em Ann. Global Anal. Geom.}, 54(3):399–447, 2018.

\bibitem{KuSt}
M.~Kunzinger and R.~Steinbauer.
\newblock Null distance and convergence of {Lorentzian} length spaces.
\newblock {\em Ann. Henri Poincar{\'e}}, 23(12):4319--4342, 2022.

\bibitem{Lev}
V.~L. Levin.
\newblock Continuous utility theorem for closed preorders on a metrizable
  {$\sigma $}-compact space.
\newblock {\em Soviet Math. Dokl.}, 28(3):715–718, 1983.

\bibitem{Li}
E.~Ling.
\newblock Aspects of {$C^0$} causal theory.
\newblock {\em Gen. Relativity Gravitation}, 52(6):Paper No. 57, 40, 2020.

\bibitem{Loll}
R.~Loll.
\newblock Quantum gravity from causal dynamical triangulations: A review.
\newblock {\em Classical Quantum Gravity}, 37(1):013002, 2020.

\bibitem{MaSi}
E.~Marczewski and R.~Sikorski.
\newblock Measures in non-separable metric spaces.
\newblock {\em Colloq. Math.}, 1:133–139, 1948.

\bibitem{McC}
R.~J. McCann.
\newblock Displacement convexity of {B}oltzmann's entropy characterizes the
  strong energy condition from general relativity.
\newblock {\em Camb. J. Math.}, 8(3):609–681, 2020.

\bibitem{Min}
E.~Minguzzi.
\newblock {$K$}-causality coincides with stable causality.
\newblock {\em Comm. Math. Phys.}, 290(1):239–248, 2009.

\bibitem{Min2}
E.~Minguzzi.
\newblock Time functions as utilities.
\newblock {\em Comm. Math. Phys.}, 298(3):855–868, 2010.

\bibitem{Min5}
E.~Minguzzi.
\newblock Topological ordered spaces as a foundation for a quantum spacetime
  theory.
\newblock {\em J. Phys. Conf. Ser.}, 442:012034, 2013.

\bibitem{Min3}
E.~Minguzzi.
\newblock The representation of spacetime through steep time functions.
\newblock {\em J. Phys. Conf. Ser.}, 968:012009, 2018.

\bibitem{Min4}
E.~Minguzzi.
\newblock {Lorentzian causality theory}.
\newblock {\em Living Rev. Rel.}, 22(1):3–204, 2019.

\bibitem{MoSu}
A.~Mondino and S.~Suhr.
\newblock An optimal transport formulation of the {Einstein} equations of
  general relativity.
\newblock {\em J. Eur. Math. Soc. (JEMS)}, 25(3):933--994, 2023.

\bibitem{Na}
L.~Nachbin.
\newblock {\em Topology and order}.
\newblock Van Nostrand Mathematical Studies, No. 4. D. Van Nostrand Co., 1965.
\newblock Translated from the Portuguese by Lulu Bechtolsheim.

\bibitem{NiPe}
H.~Nicolai and K.~Peeters.
\newblock {\em Loop and Spin Foam Quantum Gravity: A Brief Guide for
  Beginners}, page 151–184.
\newblock Springer Berlin Heidelberg, Berlin, Heidelberg, 2007.

\bibitem{NoOz}
K.~Nomizu and H.~Ozeki.
\newblock {The existence of complete Riemannian metrics}.
\newblock {\em Proc. Amer. Math. Soc.}, 12:889–891, 1961.

\bibitem{Ri}
H.~Ringstr{\"{o}}m.
\newblock {\em The {C}auchy problem in general relativity}.
\newblock ESI Lectures in Mathematics and Physics. European Mathematical
  Society (EMS), Zürich, 2009.

\bibitem{SaWu}
R.~K. Sachs and H.~H. Wu.
\newblock {\em General relativity for mathematicians}.
\newblock Graduate Texts in Mathematics, Vol. 48. Springer-Verlag, New
  York-Heidelberg, 1977.

\bibitem{Sae}
C.~S{\"{a}}mann.
\newblock Global hyperbolicity for spacetimes with continuous metrics.
\newblock {\em Ann. Henri Poincar{\'{e}}}, 17(6):1429–1455, 2016.

\bibitem{Sei}
H.-J. Seifert.
\newblock Smoothing and extending cosmic time functions.
\newblock {\em Gen. Relativity Gravitation}, 8(10):815--831, 1977.

\bibitem{Sor}
R.~D. Sorkin.
\newblock Forks in the road, on the way to quantum gravity.
\newblock {\em Int. J. Theor. Phys.}, 36:2759–2781, 1997.

\bibitem{SoWo}
R.~D. Sorkin and E.~Woolgar.
\newblock A causal order for spacetimes with {$C^0$} {L}orentzian metrics:
  proof of compactness of the space of causal curves.
\newblock {\em Classical Quantum Gravity}, 13(7):1971–1993, 1996.

\bibitem{SoVe}
C.~Sormani and C.~Vega.
\newblock Null distance on a spacetime.
\newblock {\em Classical Quantum Gravity}, 33(8):085001, 29, 2016.

\bibitem{Surya}
S.~Surya.
\newblock The causal set approach to quantum gravity.
\newblock {\em Living Rev. Rel.}, 22(1):5, 2019.

\bibitem{We}
H.~Weyl.
\newblock {\em Raum. {Z}eit. {M}aterie}, volume 251.
\newblock Springer-Verlag, Berlin, seventh edition edition, 1988.

\bibitem{Wo}
N.~M.~J. Woodhouse.
\newblock The differentiable and causal structures of space-time.
\newblock {\em J. Math. Phys.}, 14:495–501, 1973.

\end{thebibliography}

\bigskip

Annegret Burtscher and Leonardo Garc\'ia-Heveling

Department of Mathematics

Institute for Mathematics, Astrophysics and Particle Physics (IMAPP)

Radboud University Nijmegen

Postbus 9010

6500 GL Nijmegen

The Netherlands

Email: \texttt{burtscher@math.ru.nl}

\bigskip

Leonardo Garc\'ia-Heveling (current address)

Fachbereich Mathematik

Universit\"at Hamburg

Bundesstra{\ss}e 55

20146 Hamburg

Germany

Email: \texttt{leonardo.garcia@uni-hamburg.de}

\end{document}